\documentclass[a4paper,UKenglish,cleveref, autoref, thm-restate]{lipics-v2021}

\usepackage{xcolor}
\usepackage{graphicx} 
\usepackage{algorithm}
\usepackage{algpseudocode}
\usepackage{amsthm,amssymb}
\usepackage{amsmath} 
\usepackage{float}
\usepackage{comment}

\usepackage{graphicx}          
\usepackage{tikz}              
\usepackage{circuitikz}        
\usetikzlibrary{arrows.meta}   

\pdfoutput=1 
\hideLIPIcs  

\title{An Exponential Separation between Deterministic CDCL and DPLL Solvers}

\author{Sahil {Samar}}{School of Computer Science, Georgia Institute of Technology, USA}{ssamar6@gatech.edu}{0009-0000-2395-7823}{}

\author{Marc Vinyals}{School of Computer Science, University of Auckland, New Zealand}{marc.vinyals@auckland.ac.nz}{0000-0002-1487-445X}{}

\author{Vijay Ganesh}{School of Computer Science, Georgia Institute of Technology, USA}{vganesh@gatech.edu}{0000-0002-6029-2047}{}

\authorrunning{S. Samar, M. Vinyals, and V.Ganesh} 

\Copyright{Sahil Samar, Marc Vinyals, and Vijay Ganesh} 



\category{} 

\acknowledgements{We would like to thank Vincent Liew for his help in verifying the work in this paper. Additionally, we want to thank the anonymous reviewers for their valuable feedback.}

\nolinenumbers 

\EventEditors{Alexey Ignatiev and Stefan Szeider}
\EventNoEds{2}
\EventLongTitle{29th International Conference on Theory and Applications of Satisfiability Testing (SAT 2026)}
\EventShortTitle{SAT 2026}
\EventAcronym{SAT}
\EventYear{2026}
\EventDate{July 20--23, 2026}
\EventLocation{Lisbon, Portugal}
\EventLogo{}
\SeriesVolume{377}
\ArticleNo{40}

\newcommand{\IGsetup}{%
\tikzset{
  every node/.style={font=\normalsize},
  lit/.style={circle, draw, minimum size=2.5cm, inner sep=0pt, align=center},
  dot/.style={circle, draw, minimum size=0.8cm, inner sep=0pt, align=center},
  arr/.style={->, >=Stealth, line width=0.6pt, shorten >=2pt, shorten <=2pt}
}}

\begin{document}

\maketitle 
\begin{abstract}
We prove that there exists a deterministic configuration of Conflict-Driven Clause Learning (CDCL) SAT solvers using a variant of the VSIDS branching heuristic that solves instances of the Ordering Principle (OP) CNF formulas in time polynomial in $n$, where $n$ is the number of variables in such formulas. Since tree-like resolution is known to have an exponential lower bound for proof size for OP formulas, it follows that CDCL under this configuration has an exponential separation with any solver that is polynomially equivalent to tree-like resolution and therefore any configuration of DPLL SAT solvers.
\end{abstract} 

\section{Introduction}
Although SAT is the quintessential NP-complete problem, modern SAT solvers routinely handle real-world instances with millions of variables and clauses~\cite{unreasonable}. Much of this success is attributable to clause learning~\cite{grasp} and conflict-driven branching techniques such as VSIDS~\cite{chaff} and LRB~\cite{maplesat}, reinforced by extensive benchmarking and ablation studies~\cite{katebi2011empirical,elffers2018seeking}. A natural question is whether this empirical success has a theoretical explanation.

The best theoretical evidence to date comes from equivalences to proof systems: assuming optimal heuristics, DPLL is polynomially equivalent to tree-like resolution, while CDCL is polynomially equivalent to general resolution~\cite{simres,manyrestarts}. Since there exist formulas that have polynomially long resolution proofs but require exponentially long tree-like resolution proofs, as a result of the separation between proof systems we obtain an exponential separation between algorithms.

But do these results explain the efficiency of \textit{practical} solvers? The equivalence and separation results only apply to nondeterministic algorithms using optimal heuristics given by an unrealistic oracle. The fact that a short proof exists does not mean that a deterministic algorithm is necessarily able to find it. Whether that is possible is called the automatability problem~\cite{DBLP:journals/siamcomp/BonetPR00}, and it is an active area of study. In the case of resolution, a breakthrough result proved that no deterministic algorithm can always find short resolution proofs unless $\mathsf{P}=\mathsf{NP}$~\cite{DBLP:journals/jacm/AtseriasM20}.

More precisely, only one heuristic in CDCL needs to be nondeterministic for it to simulate resolution: the branching heuristic, which determines the next variable to branch on. A randomized heuristic is enough for CDCL to simulate bounded-width resolution \cite{manyrestarts}, and if the heuristic is fixed and deterministic then it is known unconditionally---independently of whether $\mathsf{P}=\mathsf{NP}$---that CDCL cannot always find short resolution proofs when the branching heuristic is VSIDS-like~\cite{pitfall} or ordered~\cite{DBLP:journals/siamcomp/MullPR22}.

This leaves open the central question we address: is there an exponential separation between DPLL and CDCL when CDCL uses a practical, fully deterministic branching heuristic?

Two families of formulas that exponentially separate tree-like and general resolution, and therefore are candidates to answer our question, are pebbling formulas and the Ordering Principle. The first family of formulas has been studied in a proof system more closely resembling CDCL~\cite{elffers2016trade}, but still assuming an optimal decision heuristic. The second family is what concerns us.

The complexity of resolution proofs of the ordering principle has been widely studied. The formulas were introduced as a \emph{tricky} family that could be solved using certain symmetry breaking properties but appeared hard for resolution~\cite{krishnamurthy1985short}. They were later shown to have in fact short resolution proofs~\cite{staalmarck1996short} but to require exponentially long tree-like resolution proofs~\cite{bonet-galesi}. Furthermore, a modified version of these formulas was used to exponentially separate regular from general resolution~\cite{alekhnovich2007exponential}.

What makes these formulas tricky is that any resolution proofs require learning clauses containing many variables, which is usually an indicator of hardness. In a well-defined sense, these formulas are just at the boundary of formulas that are easy versus hard for resolution: a formula over $n$ variables that requires clauses containing $\Omega(\sqrt{n})$ variables to prove, such as the Ordering Principle, requires exponentially long tree-like resolution proofs, and a formula whose proof requires clauses containing $\Omega(n^{1/2+\epsilon})$ variables requires exponentially long resolution proofs. This makes the Ordering Principle a very interesting case study for the capabilities of any resolution-based algorithm. Adding to the mystery, empirical work shows that the Glucose SAT solver with VSIDS decay factor $0.6$ solves OP instances in polynomial time, but not with a decay factor of $0.95$~\cite{elffers2018seeking}.

In this paper, we prove that CDCL with VSIDS solves the Ordering Principle in polynomial time. To our knowledge, this is the first result showing a polynomial upper bound for a completely deterministic configuration of CDCL solvers, simultaneously exponentially separating it from all configurations of DPLL SAT solvers. To be more precise, we use VSIDS with decay factor at most $1/2$ (which is equivalent to a variant of VMTF), fixed phase value selection, restarts after every conflict, no clause deletion, and 1UIP clause learning. 

Specifying a deterministic CDCL configuration requires a tie-breaking rule for VSIDS, and any such rule must appeal to some ordering on variables. The natural source of this ordering is the DIMACS representation of the formula, which is what real solvers receive as input. Our analysis assumes a column-major variable ordering; to extend the result to arbitrary DIMACS representations, we equip the solver with a deterministic polynomial time pre-processor that recovers this ordering, making the solver's behavior invariant to the choice of representation. This invariance is what licenses our stating the results in terms of abstract OP formulas: under it, analyzing the solver on the abstract formula is equivalent to analyzing it on every DIMACS representation.

An immediate corollary of our result, given that there is an exponential lower bound on the size of tree-like resolution proofs of the Ordering Principle, and that the DPLL algorithm is captured by tree-like resolution, is that there is an exponential separation between deterministic CDCL and nondeterministic DPLL.
\section{Preliminaries}

\subsection{Ordering Principle}

\begin{align*}
A_n &= \bigwedge_{\substack{1 \le i,j,k \le n \\ i \ne j, k \ne i, j \ne k}} (\neg P_{i,j} \lor \neg P_{j,k} \lor P_{i,k}), &&(\text{transitivity)}\\
B_n &= \bigwedge_{\substack{1 \le i,j \le n \\ i \ne j}} (\neg P_{i,j} \lor \neg P_{j,i}), && \text{(antisymmetry)}\\
D_n &= \bigwedge_{1 \le j \le n} \bigvee_{\substack{1 \le i \le n \\ i \ne j}} P_{i,j} && (\text{non-minimality})\\
A(i, j, k) &= \neg P_{i,j} \lor \neg P_{j,k} \lor P_{i,k} && (\text{clause in $A_n$})\\
B(i, j) &= \neg P_{i,j} \lor \neg P_{j,i} && (\text{clause in $B_n$})\\
D(j) &= \bigvee_{\substack{1 \le i \le n \\ i \ne j}} P_{i,j} && (\text{clause in $D_n$})
\end{align*}
\noindent An Ordering Principle instance $OP_n$ is defined as $A_n \land B_n \land D_n$. In particular, the variables in the formula are $P_{i, j}$ s.t. $1 \leq i, j \leq n$ and $i \neq j$. For a variable $P_{i, j}$, we often will call $i$ the row and $j$ the column. An $OP_n$ instance can be thought of as describing $n$ elements that are ordered in some way where $P_{i, j}$ encodes that element $i$ is smaller than element $j$, and such that the elements satisfy transitivity and antisymmetry conditions, but also with the property that no element is minimal (and thus, the formula is UNSAT).

\subsection{VSIDS}

First introduced by the authors of the Chaff solver~\cite{chaff}, VSIDS (Variable State Independent Decaying Sum) is a dynamic branching heuristic for SAT solvers that, upon learning a conflict clause, assigns a score to each variable that appears in the conflict clause (or some variation of the implication graph), decays the scores of all variables at regular intervals, and then branches on the highest scoring variable not currently on the assignment trail. There have been many variants of the heuristic since its first introduction. 

The variant of VSIDS that we consider is identical to that used in the first version of MiniSAT~\cite{minisat}, i.e., it scores \textit{variables} not literals, bumps the variables in the \textit{final} learned clause, rather than all or some subset of clauses involved in conflict analysis, and decays all variable scores by a multiplicative factor \textit{after each conflict}. There is a single hyperparameter for the heuristic, which is the decay factor $\delta$. The scores are all initialized to $0$, and updated after each conflict. 

Let the score of variable $x$ after conflict $t$ be $q(x, t)$. Then, the update formula is:
$$q(x, t) = b(x, t) + \delta q(x, t - 1)$$
where $b(x, t) = 1$ if variable $x$ participated in conflict $t$ and $b(x, t) = 0$ otherwise. A variable ``participated in a conflict'' if it was a part of the resulting learned clause of the conflict.

It is also well known that VSIDS with decay factor at most $1/2$ is equivalent to the Variable Move to Front (VMTF) branching heuristic~\cite{Rya04Efficient,DBLP:conf/sat/BiereF15}. Thus, we can replace VSIDS with VMTF for this proof, and the argument would work the same.

\subsection{CDCL Configuration}
\label{sec:configuration}

The SAT solver configuration we analyze in this paper is the following:
\begin{itemize}
    \item VSIDS Branching heuristic (as described in the previous section) with decay factor $\delta \leq 1/2$
    \item Fixed phase value selection: when making decisions, always assign variables to false.
    \item Restart after every conflict
    \item No clause deletion scheme
    \item 1UIP Clause Learning scheme
\end{itemize}

We denote the database of learned clauses by $\Gamma$. We assume that at any point during the solver run, $\Gamma$ contains all of the clauses the solver learned up to that point in the same order in which they were learned. 

We will briefly discuss the rationale behind the configuration. Our proof heavily uses the fact that the decay factor is at most $1/2$, which is perhaps not surprising in view of how an aggressive decay factor has also been experimentally observed to be required to solve instances of the Ordering Principle (as discussed in the introduction). Most assumptions we make about the configuration are inspired by practice or fairly standard in theoretical papers, except perhaps fixed phase value selection. Frequent restarts and no clause deletions are required in the existing proofs that nondeterministic CDCL simulates resolution~\cite{simres} and CDCL with randomized branching simulates bounded-width resolution~\cite{manyrestarts}, while 1UIP is the de-facto standard clause learning heuristic in modern CDCL SAT solvers. We want to point out that the separation proved in this paper may still hold without restarts; this is left to future work.

The only detail left to specify for the solver is how we break ties when there is more than one variable with the highest VSIDS score. We define the tie-break rule as follows: Let the variables $P_{i, j}$ of the OP instance be in column-major order, i.e. $$P_{2, 1}, ..., P_{n, 1}, P_{1, 2}, ..., P_{n, 2}, ..., P_{1, n}, ..., P_{n - 1, n}$$
When branching, if there is a tie for the highest VSIDS score, we break the tie based on this ordering (i.e. the variable chosen is the one with the highest score and whichever comes first in the above ordering).

The tie-breaking rule arises naturally from treating the solver's input as a DIMACS representation rather than an abstract propositional formula, a perspective we develop in Section~\ref{sec:dimacs} as the appropriate setting for analyzing deterministic solvers. Under this view, the variable ordering in the representation determines the tie-breaking rule. To make the solver's behavior invariant to the choice of representation, Section~\ref{sec:preprocessing} presents a polynomial time pre-processing algorithm that makes the solver behave identically on all valid DIMACS representations of the same formula. All results in Section~\ref{sec:main} assume this pre-processor is applied.

\subsection{DIMACS Representation}
\label{sec:dimacs}

In practice, SAT solvers take in DIMACS representations of CNF boolean formulas as input. DIMACS representations map each variable to a unique number, encode each clause as a sequence of numbers corresponding to the variables in the clause, and list each clause sequentially. Thus, a given DIMACS representation gives a formula additional structure in 3 main ways: an ordering of the variables, an ordering of the clauses, and an ordering of variables within clauses. When studying the complexity of deterministic SAT solvers, it makes sense to consider this structure that SAT solvers have access to in practice. Thus, we propose that the theoretical treatment of deterministic SAT solvers should have the input be DIMACS representations rather than arbitrary propositional formulas. Then, the deterministic solver can be analyzed with the structure that those representations impose on the formulas. In many cases, simply changing the representation of a formula can have a significant impact on the solver behavior~\cite{audemard2008experimenting}, and so we find it natural to treat solvers in this way.

For the purposes of this paper, we are primarily concerned with the VSIDS branching heuristic, and thus the aspect of the DIMACS representation that we care about is the ordering of the variables. In particular, when solving OP instances, if VSIDS scores are tied we wish to choose variables in column-major order. So, the DIMACS representation that we consider maps the variables by enumerating them in column-major order (i.e. assign $P_{2, 1}$ to $1$, $P_{3, 1}$ to $2$, and so on); any ordering of clauses and any ordering of variables within clauses is allowed. Given this representation, the column-major tie-breaking rule is simply the rule of choosing the lowest index variable -- which is the standard rule used in practice by solvers like Cadical or Kissat~\cite{DBLP:conf/sat/BiereF15}. 
Next, we discuss a way to recover the column-major variable ordering from an arbitrary DIMACS representation of OP.

\subsection{Pre-processing}
\label{sec:preprocessing}
For some classes of formulas, a polynomial time pre-processing algorithm can (using only the inherent structure of the formulas) derive enough information such that the solver behaves the same way regardless of DIMACS representation. We describe such a pre-processing algorithm for recovering the column-major variable ordering for the Ordering Principle.

Consider an arbitrary CNF formula $F$ (not necessarily an OP formula). Below we give a polytime pre-processor $P$ that takes as input $F$ and produces as output a variable ordering for $F$ with the following property: if $F$ happens to be an OP formula, then $P$ constructs a valid labeling of its variables and returns the column-major order with respect to that labeling; otherwise, $P$ returns an arbitrary order. 

\begin{enumerate}
    \item First, sort the clauses in ascending order by the number of variables in them. 
    \item Next, let the number of total variables be $N$. Solve for $n \cdot (n - 1) = N$. If $n$ is not an integer, return an arbitrary order. Otherwise, proceed. 
    \item Now, consider the $n$ largest clauses. If there are less than $n$ clauses, then return an arbitrary order. For a valid OP instance, these are the non-minimality clauses; each one corresponds to a distinct column. Prescribe to each of those $n$ clauses a unique index $j$, $1 \leq j \leq n$. Assign column index $j$ to all variables appearing in clause $j$. For an OP instance, each variable in the formula appears in exactly one of these clauses (if not, return an arbitrary order). Now, every variable has a valid column index assigned. 
    \item Next, consider the $n \cdot (n - 1) / 2$ smallest clauses. If there are not enough clauses, return an arbitrary order. For a valid OP instance, these are the antisymmetry clauses; furthermore, each variable in the formula should appear in exactly one of these clauses. Consider one of these clauses, which is of the form $\neg x_{i_1, j} \lor \neg x_{i_2, k}$, where $i_1$ and $i_2$ are the respective unknown row indices and $j$ and $k$ are the known column indices from the previous step. Now, we assign $i_1$ to be $k$ and $i_2$ to be $j$. We iterate over all of these clauses and repeat this step. If any of these clauses has more than two variables or there is a collision, return an arbitrary order. Now, for a valid OP instance, every variable has both a valid row and valid column index assigned. 
    \item Finally, now that each variable has a valid row and valid column index, we return the column-major ordering of the variables based on that labeling.
\end{enumerate}

Augmenting the CDCL solver with this pre-processor makes its behavior invariant to the choice of DIMACS representation for OP formulas: on any two representations of the same OP instance, the solver executes identically. This invariance is what justifies stating the results in Section~\ref{sec:main} in terms of abstract OP formulas rather than specific DIMACS representations. More precisely, once a deterministic solver is invariant to representation, analyzing it on the abstract formula is equivalent to analyzing it on any (and hence every) DIMACS representation of that formula.

\section{Main Result}
\label{sec:main}
We show that CDCL can determine the unsatisfiability of the Ordering Principle in polynomial time by a direct analysis of the behavior of the algorithm. That is, for each clause that the algorithm learns we show which decisions and unit propagations the solver does until reaching a conflict, building an implication graph along the way. Then we show which clause the conflict analysis heuristic learns when applied to said implication graph.

We first prove a useful invariant about the VSIDS branching heuristic under the CDCL configuration described in Section~\ref{sec:configuration}.

\begin{lemma}[{\bf Focus Lemma}]
Assume that the solver has learned at least one clause. Consider the set $S$ of variables in the most recent learned clause. During a run, when the solver has to branch, it must branch on an unassigned variable in $S$ and assign it to false, before it can branch on any other variables.
\end{lemma}

\begin{proof}
First, observe that the variables in the most recently learned clause have a VSIDS score of at least $1$. The reason is that their scores are increased by a value of $1$ after the clause was learned by definition of VSIDS. Variables not in the most recent learned clause have VSIDS score at most $\sum_{i = 1}^\infty \delta^i = \frac{\delta}{1 - \delta}$. Observe that $\frac{\delta}{1 - \delta} < 1$ for $\delta < \frac{1}{2}$. Furthermore, for any finite $k \in \mathbb{N}$, $\sum_{i = 1}^k \frac{1}{2}^i < 1$. So, the variables in the most recent learned clause all have higher VSIDS score than any variable not in the most recent learned clause for $\delta \leq 1/2$. Recall that we assume the solver restarts after every conflict. Therefore, upon branching, the solver branches on unassigned variables in the newest learned clause first, before branching on any other variables. Finally, because of fixed phase value selection, the solver assigns the (unassigned) variables in the most recent learned clause to false.
\end{proof}

Informally, what the Focus Lemma says is that whenever the solver has to make a decision during a run, the solver cannot branch on a variable which does not appear in the newest learned clause \textit{if there exists at least one unassigned variable in the newest learned clause}. Note that it could be the case that a variable in the most recent learned clause is already propagated to true before a conflict is derived, and the solver thus didn't assign false to it. The Focus Lemma says nothing about this.

Additionally, observe that the Focus Lemma can be extended beyond just the newest learned clause. Under this configuration, we can view branching as ranking each variable based solely on the most recent learned clause it appears in (where higher rank means the variable appeared in a newer learned clause), and then picking the variable with highest rank (and using the tie-breaking rule if necessary). Thus, the Focus Lemma proves the generally known equivalence between VSIDS with decay factor at most 1/2 and VMTF (while also additionally prescribing the polarity of variable assignment due to the solver configuration).  

\begin{definition}[{\bf Ordered Variable Sequence and Prefix Clause}]
For any column $j$, let $\mathcal{L}_j$ be the ordered sequence of variables in that column: $\mathcal{L}_j = \langle P_{1,j}, P_{2,j}, \dots, P_{j-1,j}, P_{j+1,j}, \dots, P_{n, j} \rangle$. Note that for $j=1$, variable $P_{1,1}$ does not exist, so the sequence starts at $P_{2,1}$. Let $C(j, k)$ be the disjunction of the first $k$ variables in $\mathcal{L}_j$. Note that $D(j) = C(j, n - 1)$.
\end{definition}

\begin{theorem}
\label{theorem:clause_structure}
For the CDCL solver configuration in Section~\ref{sec:configuration}, the solver always learns exactly the following clauses in this exact order on $OP_n$ instances for $n \geq 6$:
$$
\begin{array}{rll}
\multicolumn{3}{l}{\textbf{Head}} \\[.5em]
& C(1, n-2)  \\
& C(1, n-3) \\[1em]

\multicolumn{3}{l}{\textbf{Descending Cascade}} \\[.5em]
& \multicolumn{2}{l}{\text{For each } j \text{ descending from } n-2 \text{ to } 2:} \\[0.5em]
& \left.
    \begin{array}{l}
    C(j, n-2) \\
    C(j, n-3) \\
    \quad \vdots \\
    C(j, j-1)
    \end{array}
  \right\} & \parbox{18em}{A triangular block of depth $n-j$. \\ Starts at length $n - 2$, ends at length $j-1$.} \\[1em]

\multicolumn{3}{l}{\textbf{Tail}} \\[.5em]
& \left.
    \begin{array}{l}
    C(1, n-4) \\
    \quad \vdots \\
    C(1, 2)
    \end{array}
  \right\} & \parbox{18em}{Resumes column 1. \\ Ends with the pair $P_{2,1} \vee P_{3,1}$.}
\end{array}
$$
\end{theorem}

To illustrate the pattern, for $OP_6$ (left) and $OP_7$ (right) the learned clauses are:

\[
\makebox[\textwidth][c]{$ %
\begin{array}{@{}l@{\qquad \qquad \qquad}l@{}}
\begin{array}[t]{@{}l@{}}
P_{2,1}\;P_{3,1}\;P_{4,1}\;P_{5,1}\\
P_{2,1}\;P_{3,1}\;P_{4,1}\\[15pt]
P_{1,4}\;P_{2,4}\;P_{3,4}\;P_{5,4}\\
P_{1,4}\;P_{2,4}\;P_{3,4}\\[3pt]
P_{1,3}\;P_{2,3}\;P_{4,3}\;P_{5,3}\\
P_{1,3}\;P_{2,3}\;P_{4,3}\\
P_{1,3}\;P_{2,3}\\[3pt]
P_{1,2}\;P_{3,2}\;P_{4,2}\;P_{5,2}\\
P_{1,2}\;P_{3,2}\;P_{4,2}\\
P_{1,2}\;P_{3,2}\\
P_{1,2}\\[15pt]
P_{2,1}\;P_{3,1}
\end{array}
&
\begin{array}[t]{@{}l@{}}
P_{2,1}\;P_{3,1}\;P_{4,1}\;P_{5,1}\;P_{6,1}\\
P_{2,1}\;P_{3,1}\;P_{4,1}\;P_{5,1}\\[15pt]
P_{1,5}\;P_{2,5}\;P_{3,5}\;P_{4,5}\;P_{6,5}\\
P_{1,5}\;P_{2,5}\;P_{3,5}\;P_{4,5}\\[3pt]
P_{1,4}\;P_{2,4}\;P_{3,4}\;P_{5,4}\;P_{6,4}\\
P_{1,4}\;P_{2,4}\;P_{3,4}\;P_{5,4}\\
P_{1,4}\;P_{2,4}\;P_{3,4}\\[3pt]
P_{1,3}\;P_{2,3}\;P_{4,3}\;P_{5,3}\;P_{6,3}\\
P_{1,3}\;P_{2,3}\;P_{4,3}\;P_{5,3}\\
P_{1,3}\;P_{2,3}\;P_{4,3}\\
P_{1,3}\;P_{2,3}\\[3pt]
P_{1,2}\;P_{3,2}\;P_{4,2}\;P_{5,2}\;P_{6,2}\\
P_{1,2}\;P_{3,2}\;P_{4,2}\;P_{5,2}\\
P_{1,2}\;P_{3,2}\;P_{4,2}\\
P_{1,2}\;P_{3,2}\\
P_{1,2}\\[15pt]
P_{2,1}\;P_{3,1}\;P_{4,1}\\
P_{2,1}\;P_{3,1}
\end{array}
\end{array}
$} %
\]

\begin{proof}
The structure of our proof is to show that the solver first learns the \textbf{Head}, then transitions to and learns the \textbf{Descending Cascade}, and finally transitions to and learns the \textbf{Tail} and derives UNSAT. We inductively show that, given the solver has learned some prefix of the learned clauses in the theorem statement, it correctly learns the next clause. We show the corresponding implication graphs to learned clauses, and for each graph we show the 1UIP node and cut. First, we state a key invariant:

\begin{lemma}[{\bf Equal Score Invariant}]
For the CDCL solver configuration in Section~\ref{sec:configuration}, when solving OP instances, for each learned clause $C$ one of the following holds: (1) \textbf{all} variables in $C$ participate in a learned clause for the first time in $C$, (2) \textbf{all} variables in $C$ participate in a previously learned clause $C'$ and in no clause learned after $C'$ but before $C$. 

It follows that the variables in $C$ all have the same VSIDS score right after the clause $C$ has been learned and the scores have been bumped.
\end{lemma}

A consequence of the \textbf{Equal Score Invariant} is that we branch on variables in the most recent learned clauses (by the Focus Lemma) in column-major order (due to the tie-breaking rule). We will show that the \textbf{Equal Score Invariant} is maintained for all of the learned clauses.  

\begin{figure}[H]
\centering

\begin{minipage}[t]{0.2\textwidth}
\vspace{0pt}
\begin{enumerate}
\item[$1:$] $A(2, n, 1)$
\item[$2:$] $D(1)$
\item[$3:$] $D(1)$
\item[$4:$] $A(n - 1, n, 1)$
\item[$5:$] $A(2, n, 1)$
\item[$6:$] $A(n-1, n, 1)$
\item[$7:$] $B(1, n)$
\item[$8:$] $D(n)$
\item[$9:$] $D(n)$
\item[$10:$] $D(n)$
\end{enumerate}
\end{minipage}\hfill
\begin{minipage}[t]{0.78\textwidth}
\vspace{0pt}
\centering
\resizebox{1\textwidth}{!}{%
\begin{circuitikz}
\IGsetup

\node[lit] (0)   at (1.75,17.75) {$\neg P_{2,1}@1$};
\node[lit, draw=red, dashed] (2)  at (7.75,17.75) {$\neg P_{n-1,1}@n-2$};

\node[lit] (3)   at (4.75,14.00) {$P_{n,1}@n-2$};

\node[lit] (4)   at (1.75,10.50) {$\neg P_{2,n}@n-2$};
\node[lit] (6)  at (7.75,10.50) {$\neg P_{n-1,n}@n-2$};
\node[lit] (7)   at (12.25,10.50) {$\neg P_{1,n}@n-2$};

\node[lit] (conf)  at (7.75,6)  {CONFLICT};

\node[dot] (dotsTopR) at (4.75,17.75) {$\dots$};
\node[dot] (dotsBotM) at (4.75,10.50) {$\dots$};

\draw[arr] (0)  -- node[midway, left] {$1$} (4);
\draw[arr] (0)  -- node[midway, right] {$2$} (3);
\draw[arr] (2) -- node[midway, left] {$3$} (3);
\draw[arr] (2) -- node[midway, right] {$4$}(6);

\draw[arr] (3)  -- node[midway, left] {$5$} (4);
\draw[arr] (3)  -- node[midway, left] {$6$} (6);
\draw[arr] (3)  -- node[midway, above] {$7$} (7);

\draw[arr] (4)  -- node[midway, left] {$8$} (conf);
\draw[arr] (6) -- node[midway, left] {$9$} (conf);
\draw[arr] (7)  -- node[midway, left] {$10$} (conf);

\draw[red, dashed, line width=0.8pt] (0.75,16.25) -- (13,16.25);

\end{circuitikz}
}%
\end{minipage}

\caption{First Head Clause Implication Graph. 1UIP node and cut in red.}
\label{fig:head1}

\end{figure}
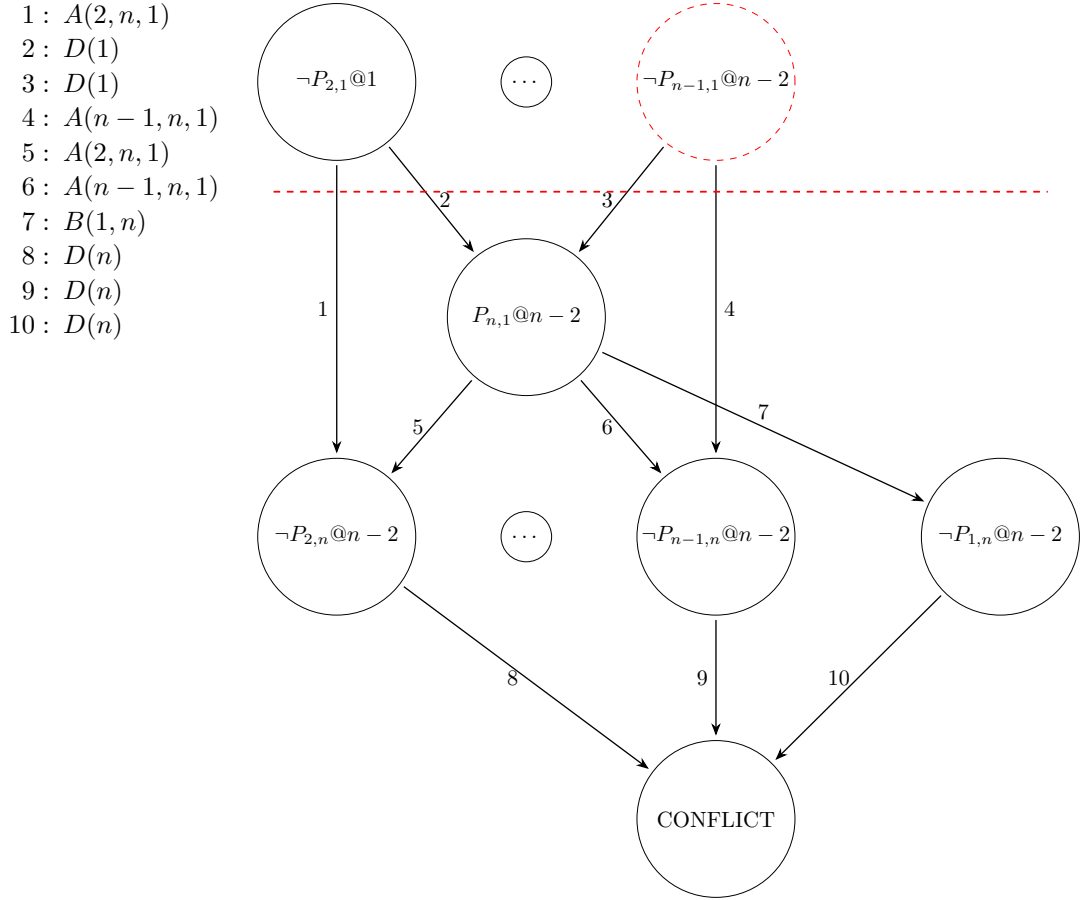

At first, since $|\Gamma| = 0$ and all VSIDS scores are $0$, the solver branches on the OP variables in column-major order; that is, the solver branches on $P_{2, 1}, P_{3, 1}, ..., P_{n - 1, 1}$ (and assigns them to false). After doing so, we get a conflict, and 1UIP derives the learned clause $C(1, n - 2)$ (Figure \ref{fig:head1}). The \textbf{Equal Score Invariant} holds by Case 1. Then, by the Focus Lemma, the solver branches on $P_{2, 1}, ..., P_{n - 2, 1}$ before any other variables. This then yields a conflict and 1UIP derives the learned clause $C(1, n - 3)$ (Figure \ref{fig:head2}). Because each variable in this learned clause was also in the previous learned clause, the \textbf{Equal Score Invariant} is maintained. So, we see that the learned clause database always starts off with \textbf{Head}, and the \textbf{Equal Score Invariant} is maintained for the \textbf{Head} by Case 2.

\begin{figure}[H]

\centering
\begin{minipage}[t]{0.25\textwidth}
\vspace{0pt}
\begin{enumerate}
\item[$1:$] $A(2,n-1,1)$
\item[$2:$] $C(1, n-2)$
\item[$3:$] $C(1, n-2)$
\item[$4:$] $A(n - 2, n - 1, 1)$
\item[$5:$] $B(1, n - 1)$
\item[$6:$] $A(2, n - 1, 1)$
\item[$7:$] $A(n - 2, n - 1, 1)$
\item[$8:$] $A(1, n, n - 1)$
\item[$9:$] $D(n - 1)$
\item[$10:$] $A(2, n, n - 1)$
\item[$11:$] $D(n - 1)$
\item[$12:$] $D(n - 1)$
\item[$13:$] $A(n - 2, n, n - 1)$
\item[$14:$] $A(1,n,n-1)$
\item[$15:$] $A(2,n,n-1)$
\item[$16:$] $A(n-2,n,n-1)$
\item[$17:$] $B(n-1,n)$
\item[$18:$] $D(n)$
\item[$19:$] $D(n)$
\item[$20:$] $D(n)$
\item[$21:$] $D(n)$
\end{enumerate}
\end{minipage}\hfill
\begin{minipage}[t]{0.75\textwidth}
\vspace{0pt}
\resizebox{1\textwidth}{!}{%
\begin{circuitikz}
\IGsetup

\node[lit] (0)   at (1.75, 24) {$\neg P_{2,1}@1$};
\node[dot] (1) at (7,24) {$\dots$};
\node[lit, draw=red, dashed] (2)  at (11.75,24) {$\neg P_{n-2,1}@n-3$};

\node[lit] (3)   at (8.00,20) {$P_{n-1,1}@n-3$};

\node[lit] (4)   at (5,16) {$\neg P_{2,n-1}@n-3$};
\node[dot] (5)   at (8,16) {$\dots$};
\node[lit] (6)  at (11,16) {$\neg P_{n-2,n-1}@n-3$};
\node[lit] (7)   at (1.75,16) {$\neg P_{1,n-1}@n-3$};

\node[lit] (8) at (9,12) {$P_{n, n - 1}@n-3$};

\node[lit] (9) at (2,8) {$\neg P_{1, n}@n-3$};
\node[lit] (10) at (5,8) {$\neg P_{2, n}@n-3$};
\node[dot] (11) at (8,8) {$\dots$};
\node[lit] (12) at (11,8) {$\neg P_{n - 2, n}@n-3$};
\node[lit] (13) at(14,8) {$\neg P_{n - 1, n}@n-3$};

\node[lit] (conf)  at (8.00,4)  {CONFLICT};

\draw[arr] (0)  -- node[midway, above] {$2$} (3);
\draw[arr] (0)  -- node[midway, left] {$1$} (4);

\draw[arr] (2) -- node[midway, above] {$3$} (3);
\draw[arr] (2) -- node[midway, left] {$4$}(6);

\draw[arr] (3)  -- node[midway, above] {$6$} (4);
\draw[arr] (3)  -- node[midway, above] {$7$} (6);
\draw[arr] (3)  -- node[midway, above] {$5$} (7);

\draw[arr] (4)  -- node[midway, right] {$11$} (8);
\draw[arr] (4)  -- node[midway, left] {$10$} (10);

\draw[arr] (6)  -- node[midway, left] {$12$} (8);
\draw[arr] (6)  -- node[midway, right] {$13$} (12);

\draw[arr] (7)  -- node[midway, above] {$9$} (8);
\draw[arr] (7)  -- node[midway, left] {$8$} (9);

\draw[arr] (8)  -- node[midway, above] {$14$} (9);
\draw[arr] (8)  -- node[midway, right] {$15$} (10);
\draw[arr] (8)  -- node[midway, left] {$16$} (12);
\draw[arr] (8)  -- node[midway, right] {$17$} (13);

\draw[arr] (9)  -- node[midway, below] {$18$} (conf);
\draw[arr] (10)  -- node[midway, right] {$19$} (conf);
\draw[arr] (12) -- node[midway, left] {$20$} (conf);
\draw[arr] (13)  -- node[midway, below] {$21$} (conf);

\draw[red, dashed, line width=0.8pt] (0.75,22) -- (14.25,22);

\end{circuitikz}
}%
\end{minipage}

\caption{Second Head Clause Implication Graph}
\label{fig:head2}
\end{figure}
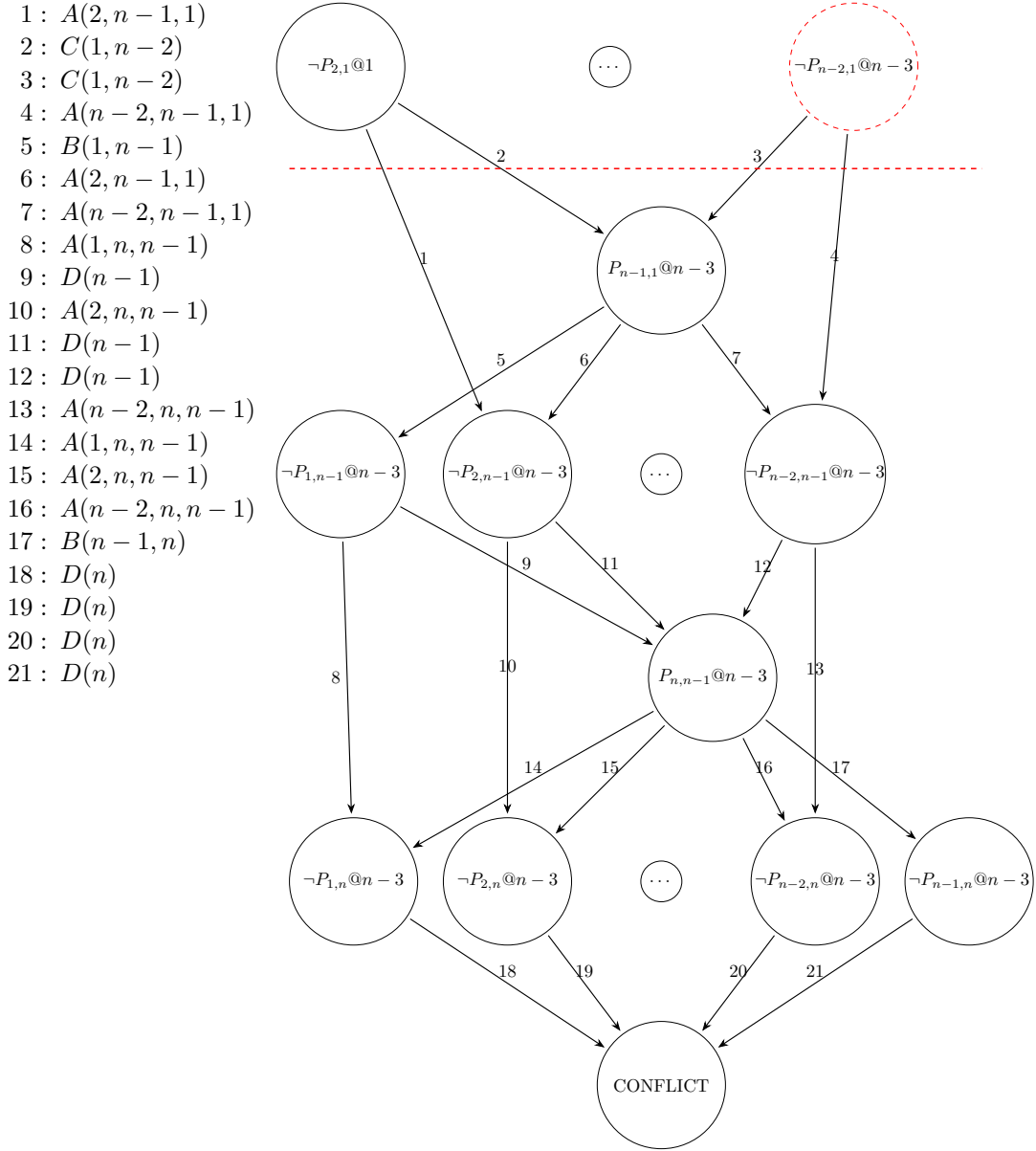

For the \textbf{Descending Cascade}, there are three facts that we show. The first is that the Descending Cascade actually {\it begins}, i.e., we learn $C(n - 2, n - 2)$ right after the head. The second is that we \textit{cascade} within a single column, i.e., given that the solver has learned $C(j, n - 2), ..., C(j, k)$ for $k > j - 1$, the solver learns $C(j, k - 1)$ next (Lemma~\ref{lem:cascade}). The third is that after learning $C(j, j - 1)$ for $j > 2$, we \textit{descend} to the next column and learn $C(j - 1, n - 2)$ (Lemma~\ref{lem:descending}).

After learning the Head, by the Focus Lemma, we see that the solver next branches on $P_{2, 1}, ..., P_{n - 3, 1}, P_{n - 1, 1}$. After doing so, 1UIP derives the learned clause $C(n - 2, n - 2)$ (Figure \ref{fig:dcas1}). None of the variables in the learned clause were involved in any of the previous learned clauses, so the \textbf{Equal Score Invariant} holds by case 1. 

\begin{figure}[H]
\centering
\begin{minipage}[t]{0.25\textwidth}
\vspace{0pt}
\begin{enumerate}
\item[$1:$] $A(2,n-2,1)$
\item[$2:$] $C(1, n - 3)$
\item[$3:$] $C(1, n-3)$
\item[$4:$] $A(n - 3, n - 2, 1)$
\item[$5:$] $A(n - 1, n - 2, 1)$
\item[$6:$] $B(1, n - 2)$
\item[$7:$] $A(2, n - 2, 1)$
\item[$8:$] $A(n - 3, n - 2, 1)$
\item[$9:$] $A(n - 1, n - 2, 1)$
\item[$10:$] $A(1, n, n - 2)$
\item[$11:$] $D(n - 2)$
\item[$12:$] $A(2, n, n - 2)$
\item[$13:$] $D(n - 2)$
\item[$14:$] $D(n - 2)$
\item[$15:$] $A(n - 3, n, n - 2)$
\item[$16:$] $D(n - 2)$
\item[$17:$] $A(n - 1, n, n - 2)$
\item[$18:$] $A(1, n, n - 2)$
\item[$19:$] $A(2, n, n - 2)$
\item[$20:$] $A(n - 3, n, n - 2)$
\item[$21:$] $A(n - 1, n, n - 2)$
\item[$22:$] $B(n - 2, n)$
\item[$23:$] $D(n)$
\item[$24:$] $D(n)$
\item[$25:$] $D(n)$
\item[$26:$] $D(n)$
\item[$27:$] $D(n)$
\end{enumerate}
\end{minipage}\hfill
\begin{minipage}[t]{0.75\textwidth}
\vspace{0pt}
\resizebox{1\textwidth}{!}{%
\begin{circuitikz}
\IGsetup

\node[lit] (0)   at (1.75,17.75) {$\neg P_{2,1}@1$};
\node[lit] (1)  at (10.50,17.75) {$\neg P_{n-3,1}@n\!-\!4$};
\node[lit] (2)  at (14.25,17.75) {$\neg P_{n-1,1}@n\!-\!3$};

\node[dot] (dotsTop1) at (3.75,17.75) {$\dots$};

\node[lit] (3)  at (8.00,14.00) {$P_{n-2,1}@n\!-\!4$};

\node[lit] (4) at (0,10.75) {$\neg P_{1,n-2}@n\!-\!4$};
\node[lit] (12) at (4.5, 10.75) {$\neg P_{2, n - 2}@n - 4$};
\node[lit] (5) at (10.25,10.75) {$\neg P_{n-3,n-2}@n\!-\!4$};
\node[lit, draw=red, dashed] (6)at (14.00,10.75) {$\neg P_{n-1,n-2}@n\!-\!3$};

\node[dot] (dotsMid1) at (7,10.75) {$\dots$};

\node[lit] (7) at (8.00,7) {$P_{n,n-2}@n\!-\!3$};

\node[lit] (8)  at (1.75,3.25) {$\neg P_{1,n}@n\!-\!3$};
\node[lit] (13) at (6, 3.25) {$\neg P_{2, n}@n - 3$};
\node[lit] (9) at (9.25,3.25) {$\neg P_{n-3,n}@n\!-\!3$};
\node[lit] (10) at (12,3.25) {$\neg P_{n-1,n}@n\!-\!3$};

\node[dot] (dotsBot1) at (3.75,3.25) {$\dots$};

\node[lit] (11) at (15,3.25) {$\neg P_{n-2,n}@n\!-\!3$};

\node[lit] (conf) at (9.00, -1) {CONFLICT};

\draw[arr] (0) -- node[midway, left] {$1$} (12);
\draw[arr] (0) -- node[midway, above] {$2$} (3);
\draw[arr] (1) -- node[midway, above] {$3$} (3);
\draw[arr] (1) -- node[midway, right] {$4$} (5);

\draw[arr] (2) -- node[midway, right] {$5$} (6);

\draw[arr] (3) -- node[midway, above] {$6$} (4);
\draw[arr] (3) -- node[midway, above] {$7$} (12);
\draw[arr] (3) -- node[midway, above] {$8$} (5);
\draw[arr] (3) -- node[midway, above] {$9$} (6);

\draw[arr] (4) -- node[midway, above] {$11$} (7);
\draw[arr] (4) -- node[midway, right] {$10$} (8);

\draw[arr] (5) -- node[midway, left] {$14$} (7);
\draw[arr] (5) -- node[midway, right] {$15$} (9);

\draw[arr] (6) -- node[midway, above] {$16$} (7);
\draw[arr] (6) -- node[midway, left] {$17$} (10);

\draw[arr] (12) -- node[midway, left] {$12$} (13);
\draw[arr] (12) -- node[midway, above] {$13$} (7);

\draw[arr] (7) -- node[midway, above] {$18$} (8);
\draw[arr] (7) -- node[midway, left] {$19$} (13);
\draw[arr] (7) -- node[midway, left] {$20$} (9);
\draw[arr] (7) -- node[midway, above] {$21$} (10);
\draw[arr] (7) -- node[midway, above] {$22$} (11);

\draw[arr] (8) -- node[midway, left] {$23$} (conf);
\draw[arr] (13) -- node[midway, left] {$24$} (conf);
\draw[arr] (9) -- node[midway, left] {$25$} (conf);
\draw[arr] (10) -- node[midway, left] {$26$} (conf);
\draw[arr] (11) -- node[midway, above] {$27$} (conf);

\draw[red, dashed, line width=0.8pt] (0,8.6) -- (15.50,8.6);

\end{circuitikz}
}%
\end{minipage}
\caption{First Descending Cascade Clause Implication Graph}
\label{fig:dcas1}
\end{figure}
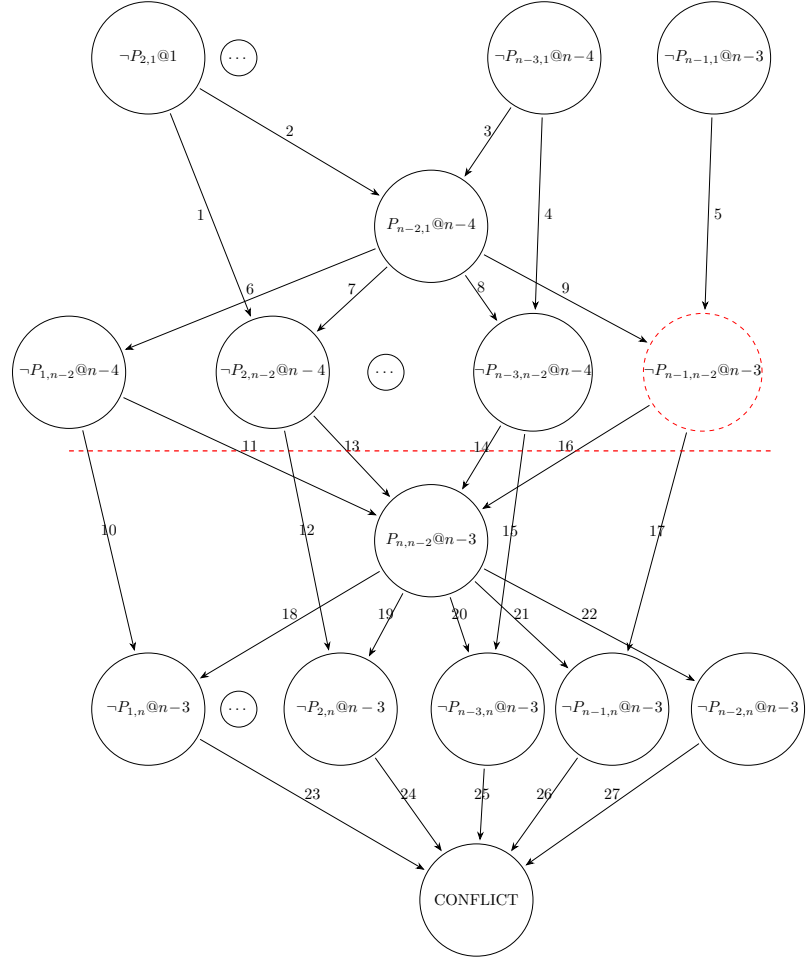
\newpage
\begin{lemma}[{\bf Cascade Lemma}] \label{lem:cascade}
Assume that $\Gamma$ contains the \textbf{Head} clauses followed by $C(j, n - 2), ..., C(j, j - 1)$ for $j = n - 2, ..., l + 1$, and then $C(l, n - 2), ..., C(l, k)$ for $k > l - 1$. Additionally, assume that the \textbf{Equal Score Invariant} holds for all of those clauses. The next learned clause is $C(l, k - 1)$, and the \textbf{Equal Score Invariant} still holds.
\end{lemma}

\noindent There are $4$ cases for the Cascade Lemma, shown below.

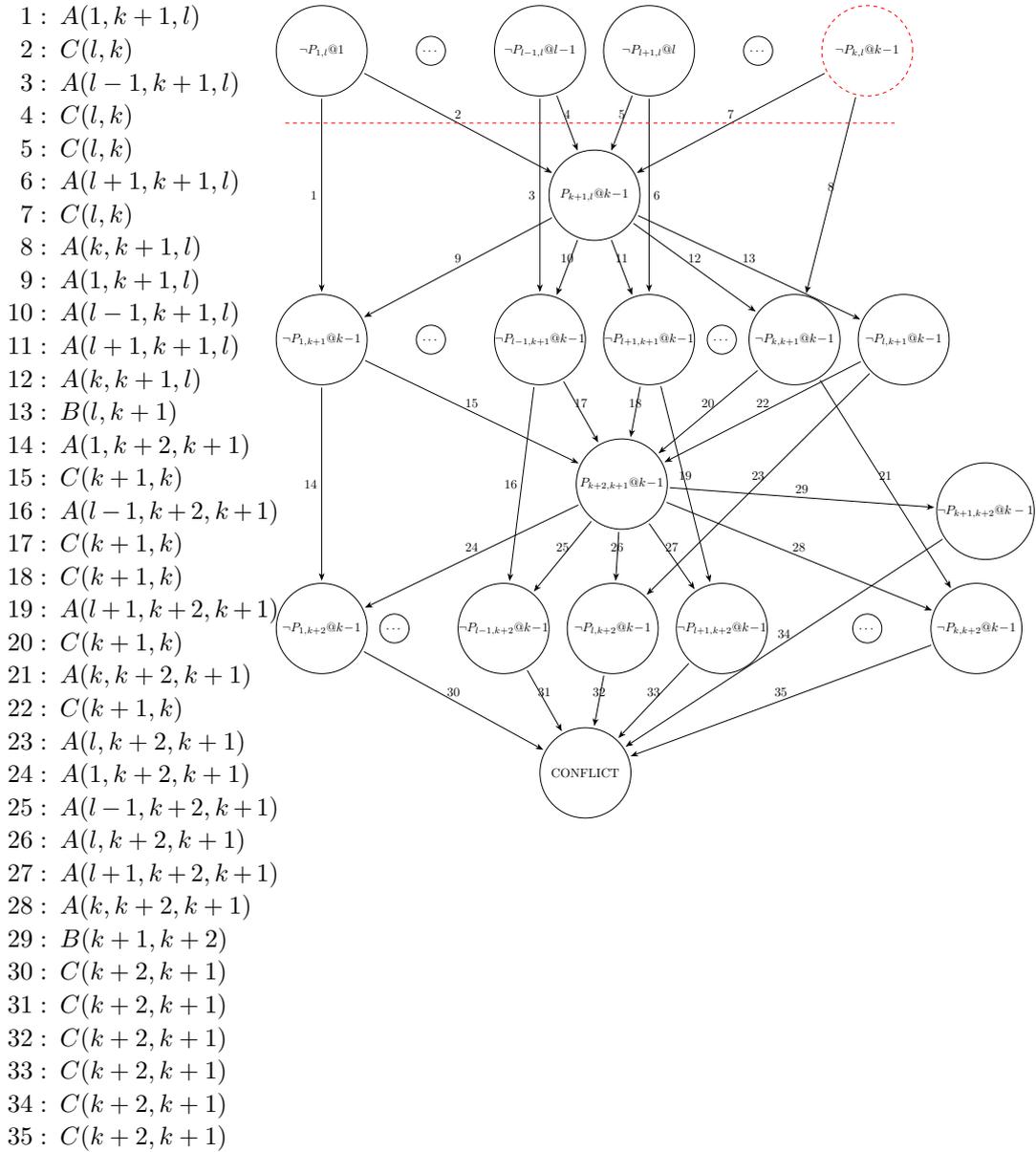
\begin{figure}[H]
\centering
\begin{minipage}[t]{0.25\textwidth}
\vspace{0pt}
\begin{enumerate}
\item[$1:$] $A(1, k + 1, l)$
\item[$2:$] $C(l, k)$
\item[$3:$] $A(l - 1, k + 1, l)$
\item[$4:$] $C(l, k)$
\item[$5:$] $C(l, k)$
\item[$6:$] $A(l + 1, k + 1, l)$
\item[$7:$] $C(l, k)$
\item[$8:$] $A(k, k + 1, l)$
\item[$9:$] $A(1, k + 1, l)$
\item[$10:$] $A(l - 1, k + 1, l)$
\item[$11:$] $A(l + 1, k + 1, l)$
\item[$12:$] $A(k, k + 1, l)$
\item[$13:$] $B(l, k + 1)$
\item[$14:$] $A(1, k + 2, k + 1)$
\item[$15:$] $C(k + 1, k)$
\item[$16:$] $A(l - 1, k + 2, k + 1)$
\item[$17:$] $C(k + 1, k)$
\item[$18:$] $C(k + 1, k)$
\item[$19:$] $A(l + 1, k + 2, k + 1)$
\item[$20:$] $C(k + 1, k)$
\item[$21:$] $A(k, k + 2, k + 1)$
\item[$22:$] $C(k + 1, k)$
\item[$23:$] $A(l, k + 2, k + 1)$
\item[$24:$] $A(1, k + 2, k + 1)$
\item[$25:$] $A(l - 1, k + 2, k + 1)$
\item[$26:$] $A(l, k + 2, k + 1)$
\item[$27:$] $A(l + 1, k + 2, k + 1)$
\item[$28:$] $A(k, k + 2, k + 1)$
\item[$29:$] $B(k + 1, k + 2)$
\item[$30:$] $C(k + 2, k + 1)$
\item[$31:$] $C(k + 2, k + 1)$
\item[$32:$] $C(k + 2, k + 1)$
\item[$33:$] $C(k + 2, k + 1)$
\item[$34:$] $C(k + 2, k + 1)$
\item[$35:$] $C(k + 2, k + 1)$
\end{enumerate}
\end{minipage}\hfill
\begin{minipage}[t]{0.75\textwidth}
\vspace{0pt}
\resizebox{1\textwidth}{!}{%
\begin{circuitikz}
\IGsetup

\node[lit] (0)   at (1.75,20) {$\neg P_{1,l}@1$};
\node[dot] (1)   at (4.75,20) {$\dots$};
\node[lit] (2)  at (7.75,20) {$\neg P_{l-1,l}@l\!-\!1$};
\node[lit] (3)  at (10.75,20) {$\neg P_{l+1,l}@l$};
\node[dot] (4)   at (13.75,20) {$\dots$};
\node[lit, draw=red, dashed] (5)  at (16.75,20) {$\neg P_{k,l}@k\!-\!1$};

\node[lit] (6)  at (9.25,16) {$P_{k+1,l}@k\!-\!1$};

\node[lit] (7)   at (1.75,12) {$\neg P_{1,k+1}@k\!-\!1$};
\node[dot] (8)   at (4.75,12) {$\dots$};
\node[lit] (9)  at (7.75,12) {$\neg P_{l-1,k+1}@k\!-\!1$};
\node[lit] (10)  at (10.75,12) {$\neg P_{l+1,k+1}@k\!-\!1$};
\node[dot] (11)   at (12.75,12) {$\dots$};
\node[lit] (12)  at (14.75,12) {$\neg P_{k,k+1}@k\!-\!1$};
\node[lit] (13)  at (17.75,12) {$\neg P_{l,k+1}@k\!-\!1$};

\node[lit] (14)  at (10,8) {$P_{k+2,k+1}@k\!-\!1$};

\node[lit] (15)   at (1.75,4) {$\neg P_{1,k+2}@k\!-\!1$};
\node[dot] (16)   at (3.75,4) {$\dots$};
\node[lit] (17)  at (6.75,4) {$\neg P_{l-1,k+2}@k\!-\!1$};
\node[lit] (18)  at (9.75,4) {$\neg P_{l,k+2}@k\!-\!1$};
\node[lit] (19)  at (12.75,4) {$\neg P_{l+1,k+2}@k\!-\!1$};
\node[dot] (20)   at (16.75,4) {$\dots$};
\node[lit] (21)  at (19.75,4) {$\neg P_{k,k+2}@k\!-\!1$};
\node[lit] (22) at (20, 7.25) {$\neg P_{k + 1, k + 2}@k - 1$};

\node[lit] (conf) at (9.00, 0) {CONFLICT};

\draw[arr] (0) -- node[midway, left] {$1$} (7);
\draw[arr] (0) -- node[midway, above] {$2$} (6);
\draw[arr] (2) -- node[midway, left] {$3$} (9);
\draw[arr] (2) -- node[midway, right] {$4$} (6);
\draw[arr] (3) -- node[midway, left] {$5$} (6);
\draw[arr] (3) -- node[midway, right] {$6$} (10);
\draw[arr] (5) -- node[midway, above] {$7$} (6);
\draw[arr] (5) -- node[midway, left] {$8$} (12);
\draw[arr] (6) -- node[midway, above] {$9$} (7);
\draw[arr] (6) -- node[midway, right] {$10$} (9);
\draw[arr] (6) -- node[midway, left] {$11$} (10);
\draw[arr] (6) -- node[midway, above] {$12$} (12);
\draw[arr] (6) -- node[midway, above] {$13$} (13);
\draw[arr] (7) -- node[midway, left] {$14$} (15);
\draw[arr] (7) -- node[midway, above] {$15$} (14);
\draw[arr] (9) -- node[midway, left] {$16$} (17);
\draw[arr] (9) -- node[midway, right] {$17$} (14);
\draw[arr] (10) -- node[midway, left] {$18$} (14);
\draw[arr] (10) -- node[midway, above] {$19$} (19);
\draw[arr] (12) -- node[midway, above] {$20$} (14);
\draw[arr] (12) -- node[midway, right] {$21$} (21);
\draw[arr] (13) -- node[midway, above] {$22$} (14);
\draw[arr] (13) -- node[midway, above] {$23$} (18);

\draw[arr] (14) -- node[midway, above] {$24$} (15);
\draw[arr] (14) -- node[midway, left] {$25$} (17);
\draw[arr] (14) -- node[midway, left] {$26$} (18);
\draw[arr] (14) -- node[midway, left] {$27$} (19);
\draw[arr] (14) -- node[midway, above] {$28$} (21);
\draw[arr] (14) -- node[midway, above] {$29$} (22);

\draw[arr] (15) -- node[midway, above] {$30$} (conf);
\draw[arr] (17) -- node[midway, right] {$31$} (conf);
\draw[arr] (18) -- node[midway, right] {$32$} (conf);
\draw[arr] (19) -- node[midway, left] {$33$} (conf);
\draw[arr] (22) -- node[midway, above] {$34$} (conf);
\draw[arr] (21) -- node[midway, above] {$35$} (conf);

\draw[red, dashed, line width=0.8pt] (0.75,17.75) -- (17.50,17.75);

\end{circuitikz}
}%
\end{minipage}

\caption{Cascade Lemma ($k \neq l, k = n - 2$) Implication Graph}
\label{fig:dcas2}

\end{figure}

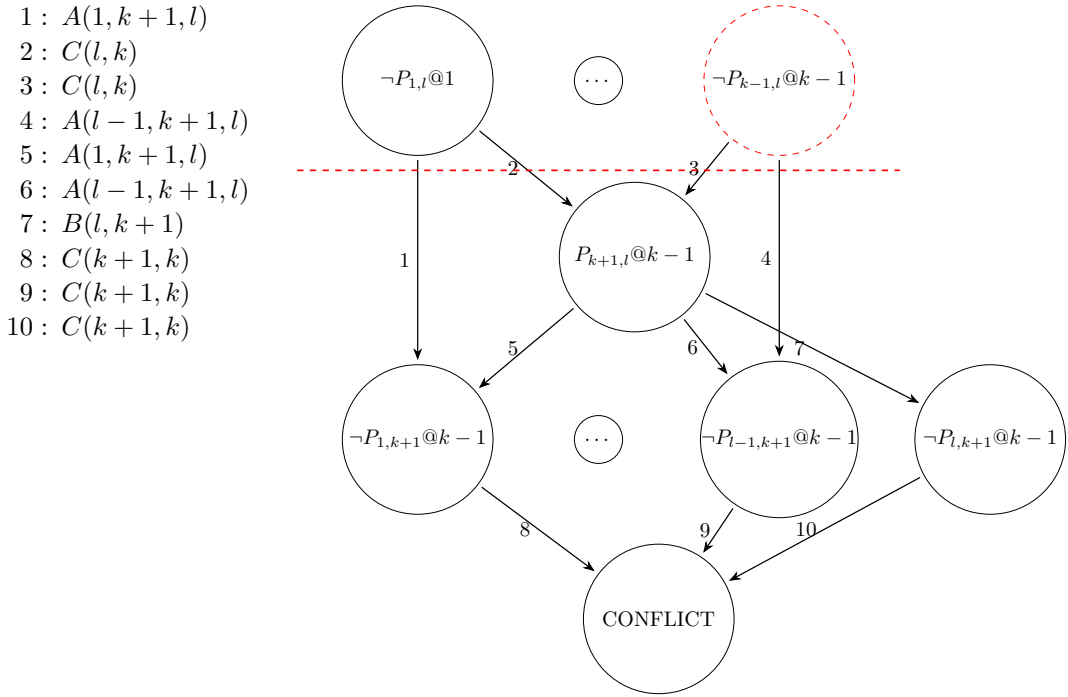
\begin{figure}[H]
\centering
\begin{minipage}[t]{0.25\textwidth}
\vspace{0pt}
\begin{enumerate}
\item[$1:$] $A(1, k + 1, l)$
\item[$2:$] $C(l, k)$
\item[$3:$] $C(l, k)$
\item[$4:$] $A(l - 1, k + 1, l)$
\item[$5:$] $A(1, k + 1, l)$
\item[$6:$] $A(l - 1, k + 1, l)$
\item[$7:$] $B(l, k + 1)$
\item[$8:$] $C(k + 1, k)$
\item[$9:$] $C(k + 1, k)$
\item[$10:$] $C(k + 1, k)$
\end{enumerate}
\end{minipage}
\begin{minipage}[t]{0.73\textwidth}
\vspace{0pt}
\resizebox{1\textwidth}{!}{%
\begin{circuitikz}
\IGsetup

\node[lit] (0)   at (2,10) {$\neg P_{1, l}@1$};
\node[dot] (1)  at (5,10) {$\dots$};
\node[lit, draw=red, dashed] (2)  at (8,10) {$\neg P_{k-1,l}@k-1$};

\node[lit] (3)  at (5.6,6.05) {$P_{k + 1,l}@k-1$};

\node[lit] (4)  at (2,3) {$\neg P_{1,k+1}@k-1$};
\node[dot] (5) at (5, 3) {$\dots$};
\node[lit] (6)  at (8,3) {$\neg P_{l-1, k+1}@k-1$};
\node[lit] (7)   at (11.5,3) {$\neg P_{l, k + 1}@k-1$};

\node[lit] (conf) at (6, 0) {CONFLICT};

\draw[arr] (0) -- node[midway, left] {$1$} (4);
\draw[arr] (0) -- node[midway, right] {$2$} (3);
\draw[arr] (2) -- node[midway, left] {$3$} (3);
\draw[arr] (2) -- node[midway, left] {$4$} (6);
\draw[arr] (3) -- node[midway, right] {$5$} (4);
\draw[arr] (3) -- node[midway, left] {$6$} (6);
\draw[arr] (3) -- node[midway, above] {$7$} (7);
\draw[arr] (4) -- node[midway, left] {$8$} (conf);
\draw[arr] (6) -- node[midway, left] {$9$} (conf);
\draw[arr] (7) -- node[midway, below] {$10$} (conf);

\draw[red, dashed, line width=0.8pt] (0,8.5) -- (10,8.5);

\end{circuitikz}
}%
\end{minipage}

\caption{Cascade Lemma ($k = l, k \neq n - 2$) Implication Graph}
\label{fig:dcas3}

\end{figure}

\begin{figure}[H]
\centering
\begin{minipage}[t]{0.25\textwidth}
\vspace{0pt}
\begin{enumerate}
\item[$1:$] $A(1, k + 1, l)$
\item[$2:$] $C(l, k)$
\item[$3:$] $A(l - 1, k + 1, l)$
\item[$4:$] $A(k - 1, k + 1, l)$
\item[$5:$] $C(l, k)$
\item[$6:$] $A(l + 1, k + 1, l)$
\item[$7:$] $C(l, k)$
\item[$8:$] $A(k, k + 1, l)$
\item[$9:$] $A(1, k + 1, l)$
\item[$10:$] $A(l - 1, k + 1, l)$
\item[$11:$] $A(l + 1, k + 1, l)$
\item[$12:$] $A(k, k + 1, l)$
\item[$13:$] $B(l, k + 1)$
\item[$14:$] $C(k + 1, k)$
\item[$15:$] $C(k + 1, k)$
\item[$16:$] $C(k + 1, k)$
\item[$17:$] $C(k + 1, k)$
\item[$18:$] $C(k + 1, k)$
\end{enumerate}
\end{minipage}
\begin{minipage}[t]{0.73\textwidth}
\vspace{0pt}
\resizebox{1\textwidth}{!}{%
\begin{circuitikz}
\IGsetup

\node[lit] (0)   at (1.75,20) {$\neg P_{1,l}@1$};
\node[dot] (1)   at (4.75,20) {$\dots$};
\node[lit] (2)  at (7.75,20) {$\neg P_{l-1,l}@l\!-\!1$};
\node[lit] (3)  at (10.75,20) {$\neg P_{l+1,l}@l$};
\node[dot] (4)   at (13.75,20) {$\dots$};
\node[lit, draw=red, dashed] (5)  at (16.75,20) {$\neg P_{k,l}@k\!-\!1$};

\node[lit] (6)  at (9.25,16) {$P_{k+1,l}@k\!-\!1$};

\node[lit] (7)   at (1.75,12) {$\neg P_{1,k+1}@k\!-\!1$};
\node[dot] (8)   at (4.75,12) {$\dots$};
\node[lit] (9)  at (7.75,12) {$\neg P_{l-1,k+1}@k\!-\!1$};
\node[lit] (10)  at (10.75,12) {$\neg P_{l+1,k+1}@k\!-\!1$};
\node[dot] (11)   at (12.75,12) {$\dots$};
\node[lit] (12)  at (14.75,12) {$\neg P_{k,k+1}@k\!-\!1$};
\node[lit] (13)  at (17.75,12) {$\neg P_{l,k+1}@k\!-\!1$};

\node[lit] (conf) at (9.00, 8) {CONFLICT};

\draw[arr] (0) -- node[midway, left] {$1$} (7);
\draw[arr] (0) -- node[midway, below] {$2$} (6);
\draw[arr] (2) -- node[midway, left] {$3$} (9);
\draw[arr] (2) -- node[midway, left] {$4$} (6);
\draw[arr] (3) -- node[midway, left] {$5$} (6);
\draw[arr] (3) -- node[midway, right] {$6$} (10);
\draw[arr] (5) -- node[midway, below] {$7$} (6);
\draw[arr] (5) -- node[midway, left] {$8$} (12);
\draw[arr] (6) -- node[midway, above] {$9$} (7);
\draw[arr] (6) -- node[midway, left] {$10$} (9);
\draw[arr] (6) -- node[midway, left] {$11$} (10);
\draw[arr] (6) -- node[midway, left] {$12$} (12);
\draw[arr] (6) -- node[midway, right] {$13$} (13);

\draw[arr] (7) -- node[midway, below] {$14$} (conf);
\draw[arr] (9) -- node[midway, left] {$15$} (conf);
\draw[arr] (10) -- node[midway, left] {$16$} (conf);
\draw[arr] (12) -- node[midway, left] {$17$} (conf);
\draw[arr] (13) -- node[midway, below] {$18$} (conf);

\draw[red, dashed, line width=0.8pt] (0.75,18.25) -- (17.50,18.25);

\end{circuitikz}
}%
\end{minipage}
\caption{Cascade Lemma $(l < k < n - 2)$ Implication Graph}
\label{fig:dcas4}

\end{figure}
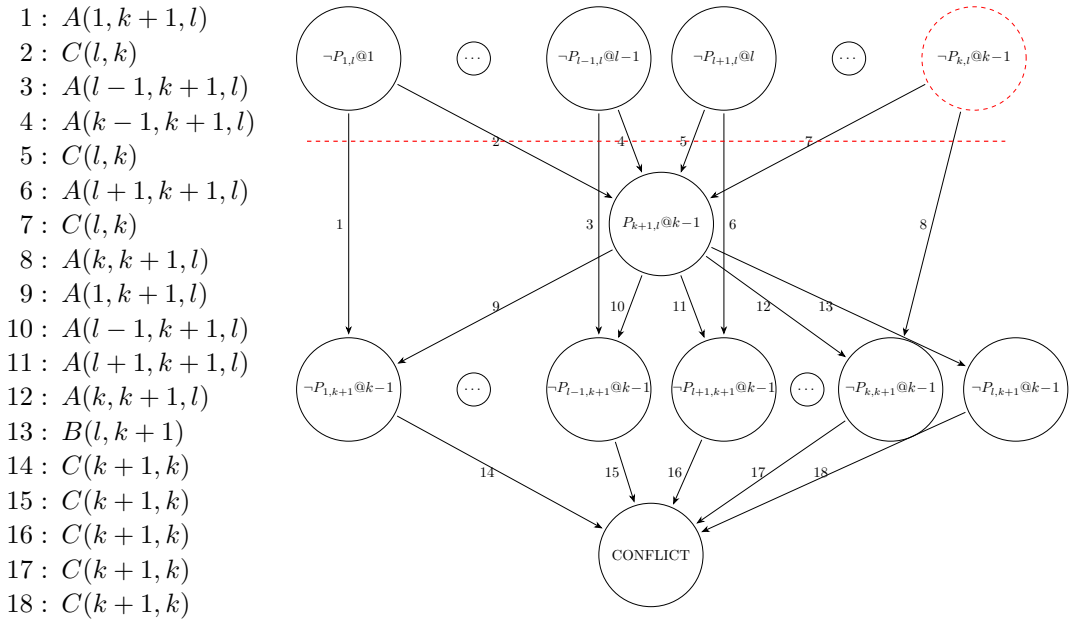

\begin{figure}[H]
\centering
\begin{minipage}[t]{0.25\textwidth}
\vspace{0pt}
\begin{enumerate}
\item[$1:$] $A(1, k+1, l)$
\item[$2:$] $C(l, k)$
\item[$3:$] $C(l, k)$
\item[$4:$] $A(k - 1, k + 1, l)$
\item[$5:$] $A(1, k+1, l)$
\item[$6:$] $A(k - 1, k + 1, l)$
\item[$7:$] $B(k, k + 1)$
\item[$8:$] $A(1, k + 2, k + 1)$
\item[$9:$] $C(k + 1, k + 1)$
\item[$10:$] $C(k + 1, k + 1)$
\item[$11:$] $A(k - 1, k + 2, k + 1)$
\item[$12:$] $C(k + 1, k + 1)$
\item[$13:$] $A(k, k + 2, k + 1)$
\item[$14:$] $A(1, k + 2, k + 1)$
\item[$15:$] $A(k - 1, k + 2, k + 1)$
\item[$16:$] $A(k, k + 2, k + 1)$
\item[$17:$] $B(k + 1, k + 2)$
\item[$18:$] $C(k + 2, k + 1)$
\item[$19:$] $C(k + 2, k + 1)$
\item[$20:$] $C(k + 2, k + 1)$
\item[$21:$] $C(k + 2, k + 1)$
\end{enumerate}
\end{minipage}
\begin{minipage}[t]{0.7\textwidth}
\vspace{0pt}
\resizebox{1\textwidth}{!}{%
\begin{circuitikz}
\IGsetup

\node[lit] (0)   at (1.75,16) {$\neg P_{1,l}@1$};
\node[dot] (1)   at (4.75,16) {$\dots$};
\node[lit, draw=red, dashed] (2)  at (7.75,16) {$\neg P_{k-1,l}@k\!-\!1$};

\node[lit] (3)  at (4.75,12) {$P_{k+1,l}@k\!-\!1$};

\node[lit] (4)   at (1.75,8) {$\neg P_{1,k+1}@k\!-\!1$};
\node[dot] (5)   at (4.75,8) {$\dots$};
\node[lit] (6)  at (7.75,8) {$\neg P_{k - 1,k+1}@k\!-\!1$};
\node[lit] (7)  at (10.75,8) {$\neg P_{k,k+1}@k\!-\!1$};

\node[lit] (8)  at (5.75,5) {$P_{k+2,k+1}@k\!-\!1$};

\node[lit] (9)   at (1.75,1) {$\neg P_{1,k+2}@k\!-\!1$};
\node[dot] (10)   at (4.75,1) {$\dots$};
\node[lit] (12) at (7.75, 1) {$\neg P_{k - 1, k + 2}@k-1$};
\node[lit] (13) at (10.75, 1) {$\neg P_{k, k + 2}@k-1$};
\node[lit] (11)  at (15.75,1) {$\neg P_{k+1,k+2}@k\!-\!1$};

\node[lit] (conf) at (5.75, -3) {CONFLICT};

\draw[arr] (0) -- node[midway, left] {$1$} (4);
\draw[arr] (0) -- node[midway, left] {$2$} (3);
\draw[arr] (2) -- node[midway, left] {$3$} (3);
\draw[arr] (2) -- node[midway, left] {$4$} (6);
\draw[arr] (3) -- node[midway, left] {$5$} (4);
\draw[arr] (3) -- node[midway, left] {$6$} (6);
\draw[arr] (3) -- node[midway, left] {$7$} (7);
\draw[arr] (4) -- node[midway, left] {$8$} (9);
\draw[arr] (4) -- node[midway, left] {$9$} (8);
\draw[arr] (6) -- node[midway, left] {$10$} (8);
\draw[arr] (7) -- node[midway, right] {$12$} (8);
\draw[arr] (6) -- node[midway, left] {$11$} (12);
\draw[arr] (7) -- node[midway, left] {$13$} (13);

\draw[arr] (8) -- node[midway, left] {$14$} (9);
\draw[arr] (8) -- node[midway, left] {$17$} (11);
\draw[arr] (8) -- node[midway, left] {$15$} (12);
\draw[arr] (8) -- node[midway, left] {$16$} (13);

\draw[arr] (9) -- node[midway, left] {$18$} (conf);
\draw[arr] (11) -- node[midway, left] {$21$} (conf);
\draw[arr] (12) -- node[midway, left] {$19$} (conf);
\draw[arr] (13) -- node[midway, left] {$20$} (conf);

\draw[red, dashed, line width=0.8pt] (0.75,14.25) -- (10,14.25);

\end{circuitikz}
}%
\end{minipage}
\caption{Cascade Lemma ($k = l = n - 2$) Implication Graph}
\label{fig:dcas5}

\end{figure}
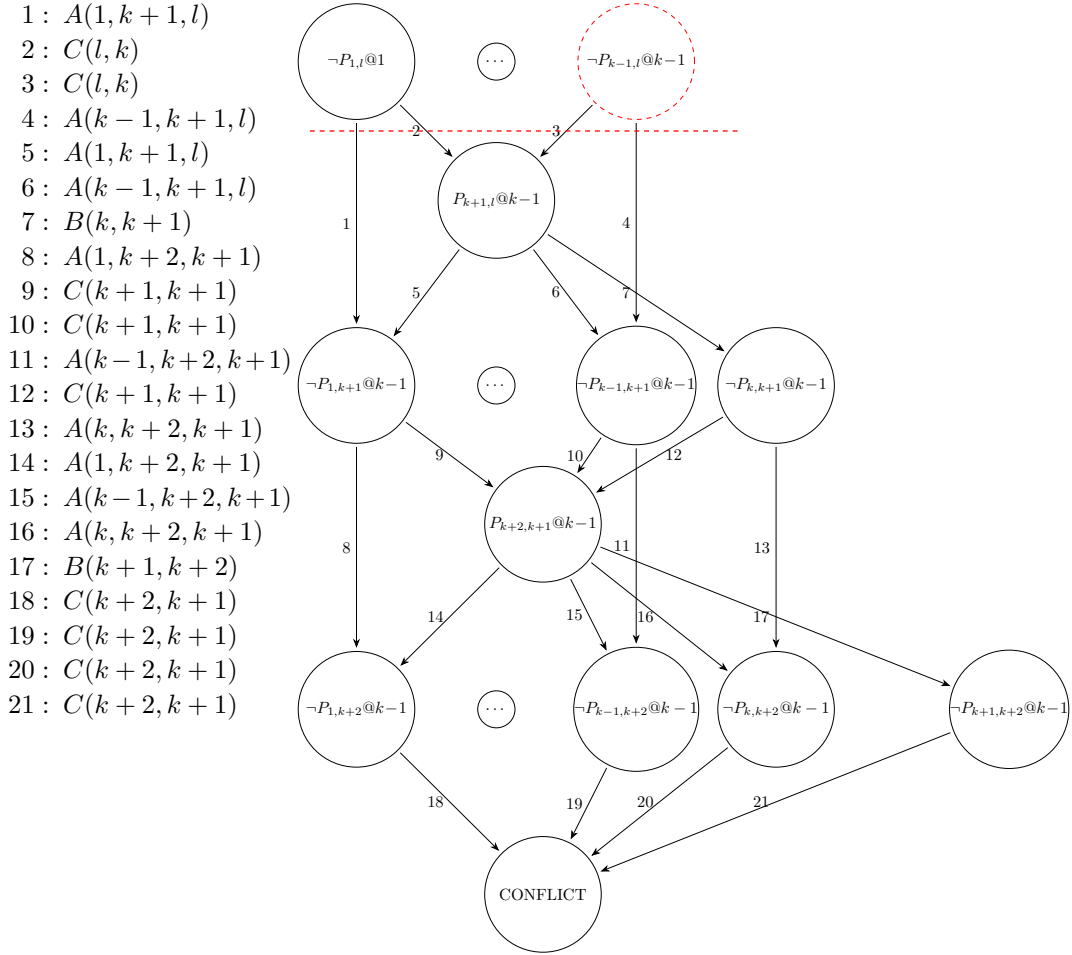

\noindent The \textbf{Equal Score Invariant} holds by Case 2 for the Cascade Lemma.
\begin{lemma}[{\bf Descending Lemma}] \label{lem:descending}
Assume that $\Gamma$ contains the \textbf{Head} clauses followed by $C(j, n - 2), ..., C(j, j - 1)$ for $j = n - 2, ..., l$, for $l > 2$. Additionally, assume that the \textbf{Equal Score Invariant} holds for all of those clauses. Then, the next learned clause is $C(l - 1, n - 2)$, and the \textbf{Equal Score Invariant} still holds.
\end{lemma}

\begin{figure}[H]
\label{fig:descending}
\centering
\resizebox{0.95\textwidth}{!}{%
\begin{circuitikz}
\IGsetup

\node[lit] (0)   at (1.75,32) {$\neg P_{1,l}@1$};
\node[dot] (1)   at (5.50,32) {$\dots$};
\node[lit] (2)  at (10.50,32) {$\neg P_{l-2,l}@l\!-\!2$};

\node[lit] (3)  at (12.5,28) {$P_{l-1,l}@l\!-\!2$};

\node[lit] (4) at (1.75,24) {$\neg P_{1,l-1}@l\!-\!2$};
\node[dot] (5) at (4.75,24) {$\dots$};
\node[lit] (6) at (7.75,24) {$\neg P_{l - 2,l - 1}@l\!-\!2$};
\node[lit] (7) at (12.5,24) {$\neg P_{l,l-1}@l\!-\!2$};

\node[lit] (8)  at (17.5,20) {$\neg P_{l+1,l}@l\!-\!1$};
\node[dot] (9)  at (20,20) {$\dots$};
\node[lit] (10) at (25.5,20) {$\neg P_{n-1,l}@n\!-\!3$};

\node[lit] (11)  at (17.5,16) {$\neg P_{l+1,l-1}@l\!-\!1$};
\node[dot] (12)  at (20,16) {$\dots$};
\node[lit, draw=red, dashed] (13) at (25.5,16) {$\neg P_{n-1,l-1}@n\!-\!3$};

\node[lit] (14) at (10.75,12) {$P_{n,l-1}@n\!-\!3$};

\node[lit] (15) at (2,7) {$\neg P_{1,n}@n\!-\!3$};
\node[dot] (16) at (4.5, 7) {$\dots$};
\node[lit] (17) at (8.5,7) {$\neg P_{l-2,n}@n\!-\!3$};
\node[lit] (18) at (12.5,7) {$\neg P_{l,n}@n\!-\!3$};
\node[lit] (19) at (15.5,7) {$\neg P_{l+1,n}@n\!-\!3$};
\node[dot] (20) at (17.5, 7) {$\dots$};
\node[lit] (21) at (19.5,7) {$\neg P_{n-1,n}@n\!-\!3$};
\node[lit] (22) at (24.5,7) {$\neg P_{l-1,n}@n\!-\!3$};

\node[lit] (23) at (0, 10.5) {$P_{n, l}@n-3$};

\node[lit] (conf) at (14, 2) {CONFLICT};

\draw[arr] (0) -- node[midway, left] {$1$} (4);
\draw[arr] (0) -- node[midway, above] {$2$} (3);
\draw[arr] (2) -- node[midway, left] {$3$} (6);
\draw[arr] (2) -- node[midway, left] {$4$} (3);
\draw[arr] (3) -- node[midway, above] {$5$} (4);
\draw[arr] (3) -- node[midway, left] {$6$} (6);
\draw[arr] (3) -- node[midway, left] {$7$} (23);
\draw[arr] (3) -- node[midway, left] {$8$} (7);
\draw[arr] (3) -- node[midway, left] {$9$} (11);
\draw[arr] (3) -- node[midway, left] {$10$} (13);
\draw[arr] (4) -- node[midway, left] {$11$} (15);
\draw[arr] (4) -- node[midway, left] {$12$} (14);
\draw[arr] (6) -- node[midway, left] {$13$} (17);
\draw[arr] (6) -- node[midway, left] {$14$} (14);
\draw[arr] (7) -- node[midway, left] {$15$} (14);
\draw[arr] (7) -- node[midway, left] {$16$} (18);
\draw[arr] (8) -- node[midway, left] {$17$} (11);
\draw[arr] (10) -- node[midway, left] {$18$} (13);
\draw[arr] (11) -- node[midway, below] {$19$} (14);
\draw[arr] (11) -- node[midway, left] {$20$} (19);
\draw[arr] (13) -- node[midway, below] {$21$} (14);
\draw[arr] (13) -- node[midway, left] {$22$} (21);
\draw[arr] (14) -- node[midway, above] {$23$} (23);
\draw[arr] (14) -- node[midway, left] {$24$} (15);
\draw[arr] (14) -- node[midway, left] {$25$} (17);
\draw[arr] (14) -- node[midway, left] {$26$} (18);
\draw[arr] (14) -- node[midway, left] {$27$} (19);
\draw[arr] (14) -- node[midway, left] {$28$} (21);
\draw[arr] (14) -- node[midway, left] {$29$} (22);

\draw[arr] (15) -- node[midway, below] {$30$} (conf);
\draw[arr] (17) -- node[midway, left] {$31$} (conf);
\draw[arr] (18) -- node[midway, left] {$32$} (conf);
\draw[arr] (19) -- node[midway, left] {$33$} (conf);
\draw[arr] (21) -- node[midway, left] {$34$} (conf);
\draw[arr] (22) -- node[midway, left] {$35$} (conf);

\draw[red, dashed, line width=0.8pt] (0.75,14.25) -- (26,14.25);

\end{circuitikz}
}%
\caption{Descending Lemma Implication Graph}
\label{fig:dcas6}
\end{figure}

\begin{minipage}[t]{0.3\textwidth}
\begin{enumerate}
\item[$1:$] $A(1, l - 1, l)$
\item[$2:$] $C(l, l - 1)$
\item[$3:$] $A(l - 2, l - 1, l)$
\item[$4:$] $C(l, l - 1)$
\item[$5:$] $A(1, l - 1, l)$
\item[$6:$] $A(l - 2, l - 1, l)$
\item[$7:$] $A(n, l - 1, l)$
\item[$8:$] $B(l - 1, l)$
\item[$9:$] $A(l + 1, l - 1, l)$
\item[$10:$] $A(n - 1, l - 1, l)$
\item[$11:$] $A(1, n, l - 1)$
\item[$12:$] $D(l - 1)$
\end{enumerate}
\end{minipage}
\hfill
\begin{minipage}[t]{0.3\textwidth}
\begin{enumerate}
\item[$13:$] $A(l - 2, n, l - 1)$
\item[$14:$] $D(l - 1)$
\item[$15:$] $D(l - 1)$
\item[$16:$] $A(l, n, l - 1)$
\item[$17:$] $A(l + 1, l - 1, l)$
\item[$18:$] $A(n - 1, l - 1, l)$
\item[$19:$] $D(l - 1)$
\item[$20:$] $A(l + 1, n, l - 1)$
\item[$21:$] $D(l - 1)$
\item[$22:$] $A(n - 1, n, l - 1)$
\item[$23:$] $A(n, l - 1, l)$
\item[$24:$] $A(1, n, l - 1)$
\end{enumerate}
\end{minipage}
\begin{minipage}[t]{0.3\textwidth}
\begin{enumerate}
\item[$25:$] $A(l - 2, n, l - 1)$
\item[$26:$] $A(l, n, l - 1)$
\item[$27:$] $A(l + 1, n, l - 1)$
\item[$28:$] $A(n - 1, n, l - 1)$
\item[$29:$] $B(l - 1, n)$
\item[$30:$] $D(n)$
\item[$31:$] $D(n)$
\item[$32:$] $D(n)$
\item[$33:$] $D(n)$
\item[$34:$] $D(n)$
\item[$35:$] $D(n)$
\end{enumerate}
\end{minipage}
\noindent The \textbf{Equal Score Invariant} holds by Case 1 for the Descending Lemma.

We now move on to the \textbf{Tail}. We first notice that the last learned clause in the Descending Cascade is always the unit clause $C(2, 1) = P_{1, 2}$. There are two things we need to show. The first is that the Tail begins, i.e., we learn $C(1, n - 4)$ after learning $P_{1, 2}$ (Lemma \ref{lem:pivot}). The second is that after learning $C(1, n - 4), ..., C(1, k)$ for $k > 2$, the next learned clause is $C(1, k - 1)$ (Lemma \ref{lem:tail}). 

\begin{lemma}[{\bf Pivot Lemma}] \label{lem:pivot}
Assume that $\Gamma$ contains the \textbf{Head} clauses followed by the \textbf{Descending Cascade} clauses, the last of which is the unit clause $P_{1, 2}$. Additionally, assume that the \textbf{Equal Score Invariant} holds for all of those clauses. The next learned clause is $C(1, n - 4)$, and the \textbf{Equal Score Invariant} still holds.
\end{lemma}

\begin{figure}[H]
\centering
\begin{minipage}[t]{0.25\textwidth}
\vspace{0pt}
\begin{enumerate}
\item[$1:$] $B(1, 2)$
\item[$2:$] $A(3, 1, 2)$
\item[$3:$] $A(n - 3, 1, 2)$
\item[$4:$] $A(3, 1, 2)$
\item[$5:$] $A(n - 3, 1, 2)$
\item[$6:$] $A(2, n - 2, 1)$
\item[$7:$] $C(1, n - 3)$
\item[$8:$] $A(3, n - 2, 1)$
\item[$9:$] $C(1, n - 3)$
\item[$10:$] $C(1, n - 3)$
\item[$11:$] $A(n - 3, n - 2, 1)$
\item[$12:$] $A(2, n - 2, 1)$
\item[$13:$] $A(3, n - 2, 1)$
\item[$14:$] $A(n - 3, n - 2, 1)$
\item[$15:$] $B(1, n - 2)$
\item[$16:$] $A(l, n, l - 1)$
\item[$17:$] $A(l + 1, l - 1, l)$
\item[$18:$] $A(n - 1, l - 1, l)$
\item[$19:$] $C(n - 2, n - 3)$
\end{enumerate}
\end{minipage}\hfill
\begin{minipage}[t]{0.75\textwidth}
\vspace{0pt}
\resizebox{1\textwidth}{!}{%
\begin{circuitikz}
\IGsetup

\node[lit] (0)   at (1.75,16) {$P_{1,2}@0$};
\node[lit] (1)  at (5.75,16) {$\neg P_{3,2}@1$};
\node[dot] (2)   at (8.75,16) {$\dots$};
\node[lit] (3)  at (11.75,16) {$\neg P_{n-3,2}@n\!-\!5$};

\node[lit] (4)  at (1,12) {$\neg P_{2,1}@0$};
\node[lit] (5)  at (4.75,12) {$\neg P_{3,1}@1$};
\node[dot] (6)   at (7.75,12) {$\dots$};
\node[lit, draw=red, dashed] (7)  at (10.75,12) {$\neg P_{n-3,1}@n\!-\!5$};

\node[lit] (8)  at (10,8) {$P_{n-2,1}@n\!-\!5$};

\node[lit] (9)  at (3.75,4) {$\neg P_{2,n-2}@n\!-\!5$};
\node[lit] (10)  at (6.75,4) {$\neg P_{3,n-2}@n\!-\!5$};
\node[dot] (11)   at (9.25,4) {$\dots$};
\node[lit] (12)  at (12.75,4) {$\neg P_{n-3,n-2}@n\!-\!5$};
\node[lit] (13)  at (15.75,4) {$\neg P_{1,n-2}@n\!-\!5$};

\node[lit] (conf) at (9.00, 0) {CONFLICT};

\draw[arr] (0) -- node[midway, left] {$1$} (4);
\draw[arr] (0) -- node[midway, left] {$2$} (5);
\draw[arr] (0) -- node[midway, above] {$3$} (7);
\draw[arr] (1) -- node[midway, left] {$4$} (5);
\draw[arr] (3) -- node[midway, left] {$5$} (7);
\draw[arr] (4) -- node[midway, left] {$6$} (9);
\draw[arr] (4) -- node[midway, below] {$7$} (8);
\draw[arr] (5) -- node[midway, left] {$8$} (10);
\draw[arr] (5) -- node[midway, right] {$9$} (8);
\draw[arr] (7) -- node[midway, left] {$10$} (8);
\draw[arr] (7) -- node[midway, right] {$11$} (12);
\draw[arr] (8) -- node[midway, left] {$12$} (9);
\draw[arr] (8) -- node[midway, left] {$13$} (10);
\draw[arr] (8) -- node[midway, left] {$14$} (12);
\draw[arr] (8) -- node[midway, right] {$15$} (13);
\draw[arr] (9) -- node[midway, left] {$16$} (conf);
\draw[arr] (10) -- node[midway, left] {$17$} (conf);
\draw[arr] (12) -- node[midway, left] {$18$} (conf);
\draw[arr] (13) -- node[midway, left] {$19$} (conf);

\draw[red, dashed, line width=0.8pt] (0.75,10.4) -- (15.50,10.4);

\end{circuitikz}
}%
\end{minipage}

\caption{Pivot Lemma Implication Graph}
\label{fig:pivot}
\end{figure}

\noindent The \textbf{Equal Score Invariant} holds by Case 2 for the Pivot Lemma.

\begin{lemma}[{\bf Tail Lemma}] \label{lem:tail}
Assume that $\Gamma$ contains the \textbf{Head clauses} followed by the \textbf{Descending Cascade} clauses followed by the clauses $C(1, n - 4), ..., C(1, k)$ for $k > 2$. Additionally, assume that the \textbf{Equal Score Invariant} holds for all of those clauses. The next learned clause is $C(1, k - 1)$, and the \textbf{Equal Score Invariant} still holds.
\end{lemma}

\begin{figure}[H]
\centering
\begin{minipage}[t]{0.25\textwidth}
\vspace{0pt}
\begin{enumerate}
\item[$1:$] $A(k + 1, 1, 2)$
\item[$2:$] $B(1, 2)$
\item[$3:$] $A(2, 1, k + 1)$
\item[$4:$] $C(1, k)$
\item[$5:$] $A(3, k + 1, 1)$
\item[$6:$] $C(1, k)$
\item[$7:$] $C(1, k)$
\item[$8:$] $A(k, k + 1, 1)$
\item[$9:$] $A(k + 1, 1, 2)$
\item[$10:$] $A(2, 1, k + 1)$
\item[$11:$] $A(3, k + 1, 1)$
\item[$12:$] $A(k, k + 1, 1)$
\item[$13:$] $B(1, k + 1)$
\item[$14:$] $C(k, k + 1)$
\item[$15:$] $C(k, k + 1)$
\item[$16:$] $C(k, k + 1)$
\item[$17:$] $C(k, k + 1)$
\end{enumerate}
\end{minipage}\hfill
\begin{minipage}[t]{0.75\textwidth}
\vspace{0pt}
\resizebox{1\textwidth}{!}{%
\begin{circuitikz}
\IGsetup

\node[lit] (0)   at (1.75,12) {$P_{1,2}@0$};
\node[lit] (4)  at (5.75,12) {$\neg P_{2,1}@0$};
\node[lit] (1)  at (8.75,12) {$\neg P_{3,1}@1$};
\node[dot] (2)   at (11.75,12) {$\dots$};
\node[lit, draw=red, dashed] (3)  at (14.75,12) {$\neg P_{k,1}@k\!-\!2$};

\node[lit] (5)  at (9.25,8) {$P_{k + 1,1}@k-2$};

\node[lit] (6)  at (3.75,4) {$\neg P_{2,k+1}@k\!-\!2$};
\node[lit] (7) at (6.75, 4) {$\neg P_{3,k+1}@k\!-\!2$};
\node[dot] (8)  at (9.75,4) {$\dots$};
\node[lit] (9)   at (12.25,4) {$\neg P_{k, k + 1}@k\!-\!2$};
\node[lit] (10)  at (15.75,4) {$\neg P_{1,k+1}@k\!-\!2$};

\node[lit] (11) at (0, 4) {$P_{k + 1, 2}@k - 2$};

\node[lit] (conf) at (9.00, 0) {CONFLICT};

\draw[arr] (0) -- node[midway, left] {$1$} (11);
\draw[arr] (0) -- node[midway, above] {$2$} (4);
\draw[arr] (1) -- node[midway, left] {$6$} (5);
\draw[arr] (1) -- node[midway, left] {$5$} (7);
\draw[arr] (3) -- node[midway, left] {$7$} (5);
\draw[arr] (3) -- node[midway, left] {$8$} (9);
\draw[arr] (5) -- node[midway, below] {$9$} (11);
\draw[arr] (4) -- node[midway, left] {$4$} (5);
\draw[arr] (4) -- node[midway, left] {$3$} (6);
\draw[arr] (5) -- node[midway, left] {$10$} (6);
\draw[arr] (5) -- node[midway, left] {$11$} (7);
\draw[arr] (5) -- node[midway, left] {$12$} (9);
\draw[arr] (5) -- node[midway, left] {$13$} (10);
\draw[arr] (6) -- node[midway, left] {$14$} (conf);
\draw[arr] (7) -- node[midway, left] {$15$} (conf);
\draw[arr] (9) -- node[midway, left] {$16$} (conf);
\draw[arr] (10) -- node[midway, left] {$17$} (conf);

\draw[red, dashed, line width=0.8pt] (0.75,10.25) -- (15.50,10.25);

\end{circuitikz}
}%
\end{minipage}

\caption{Tail Lemma Implication Graph}
\label{fig:tail}
\end{figure}

\noindent The \textbf{Equal Score Invariant} holds by Case 2 for the Tail Lemma. 

The lemmas are enough to show that we learn the \textbf{Head}, \textbf{Descending Cascade}, and \textbf{Tail} clauses. To complete the proof of the theorem, we show that the CDCL solver derives UNSAT upon learning these clauses. As part of the \textbf{Descending Cascade}, the unit clause $P_{1, 2}$ is learnt. Propagating it, we derive $\neg P_{2, 1}$. Then, the clause $P_{2, 1} \lor P_{3, 1}$, which is in the \textbf{Tail}, becomes the unit clause $P_{3, 1}$. Then the transitivity clause $A(2, 3, 1)$ becomes the unit clause $\neg P_{2, 3}$. And, the antisymmetry clause $B(1, 3)$ becomes the unit clause $\neg P_{1, 3}$. We then see that the clause $P_{1, 3} \lor P_{2, 3}$, which is part of the \textbf{Descending Cascade}, is completely falsified. However, we did not have to make any decisions to arrive at this conflict, and therefore we derive UNSAT.
\end{proof}

A natural question is whether we crucially need tie-breaking throughout the proof. To answer this question we might try to generalize the proof in a way so that the order in which we branch does not matter, as long as our branching is column-focused (i.e. always branch in the same column as the previous decision if possible). This would result in perhaps processing columns in a different order but the proof would, plausibly, be equivalent up to symmetries. More precisely, we can assume that the first learned clause is $C(1, n - 2)$, i.e. that we guarantee that we start off column-focused; then, we ask if the Focus Lemma and the rest of the configuration would result in the solver learning the essentially the same proof without further invoking the tie-breaking rule. The answer is negative due to the Descending Lemma, the mechanism by which we transition columns in the Descending Cascade, which requires that we initially propagate $P_{l - 1, l}$ from the first $l - 2$ decisions if we want to learn a clause in column $l - 1$. So, the previous learned clauses cannot consist of arbitrary subsets of variables from column $l$ if we want to descend to column $l - 1$.
This suggests that tie-breaking does play an important role in this particular proof, and that further generalizations that removed tie-breaking would be more complex if not longer.

\newpage

\begin{theorem}
The total number of conflicts for the deterministic configuration of CDCL SAT solvers in Section~\ref{sec:configuration} to derive UNSAT for an $OP_n$ instance is exactly $\binom{n}{2} - 3$, for $n \geq 6$.
\end{theorem}

\begin{proof}
    The \textbf{Head} has $2$ clauses, the \textbf{Descending Cascade} has $\sum_{j = 2}^{n - 2} n - j = n (n - 3) / 2$ clauses, and the \textbf{Tail} has $n - 5$ clauses. In total, then, there are $\frac{(n - 3) (n + 2)}{2} = \binom{n}{2} - 3$ clauses. 
\end{proof}

\begin{corollary}
There exists a deterministic CDCL solver that has an exponential separation with any solver that is polynomially equivalent to tree-like resolution (e.g. DPLL with nondeterministic branching). 
\end{corollary}

\begin{proof}
It is known that tree-like resolution has an exponential lower bound for proof size for Ordering Principle instances~\cite{bonet-galesi}. We have shown the CDCL configuration described in Section~\ref{sec:configuration} requires exactly $\binom{n}{2} - 3$ conflicts to derive UNSAT for $OP_n$ instances. It is well known that a polynomial number of conflicts implies polynomial runtime for CDCL SAT solvers. It then follows that this CDCL solver has an exponential separation with any solver that is polynomially equivalent to tree-like resolution.
\end{proof}

\begin{remark}
A separation in the opposite direction also holds: there exists a class of formulas which can be solved efficiently by nondeterministic DPLL but not by VSIDS-like CDCL. This follows by noticing that the pitfall formulas that separate VSIDS-like CDCL from resolution~\cite{pitfall} have short proofs in tree-like resolution.
\end{remark}

\section{Correctness}

For proofs of correctness for the implication graphs, see the Appendix.

We also experimentally verified our results using the ``textbooksat'' SAT solver ~\cite{vinyals_textbooksat} on instances up to $OP_{30}$, analyzing the exact learned clause structures and propagations on the same solver configuration. Our results were confirmed on all of those instances.

\section{Conclusions}
In this paper, we establish an exponential separation between a deterministic configuration of CDCL SAT solvers and any configuration of DPLL SAT solvers. Specifically, we show that the proposed CDCL configuration solves the OP class of formulas in polynomial time, even though tree-like resolution is known to admit only exponential-size proofs for these instances.

We hope this work proves useful for the future analysis of practical, deterministic SAT solvers. Our argument is considerably simplified by two key invariants (the Focus Lemma and the Equal Score Invariant) and we believe that identifying such invariants is essential for making deterministic analysis tractable in general. Moreover, treating the input to a deterministic solver as a DIMACS representation, rather than as an abstract propositional formula, helps bridge the gap between theory and practice: it allows structural information about the formula that is available in empirical settings to be specified precisely and reasoned about formally. Pairing this perspective with deterministic polynomial time pre-processing offers a natural way to recover invariance to the choice of DIMACS representation, and we expect it to be a useful tool in subsequent work.

\newpage 

\bibliography{references} 

\newpage

\section{Appendix}
To complete the proof of Theorem \ref{theorem:clause_structure}, we show that each of the implication graphs corresponding to the lemmas is \textit{complete}.

\begin{definition}[\bf Completeness]
We say that an implication graph is complete if it is impossible for a different conflict to be learned under the same decisions. 
\end{definition}

We use the following conventions. A literal is \textit{derived} if it is either \textit{decided} (the solver branches on the variable and assigns it a value) or if it is \textit{implied} by other literals. A literal $l$ can only be \textit{implied} when, for some clause $C$ s.t. $l \in C$, all literals except $l$ have their negation derived. In that case, all literals in $C \setminus l$ together imply $l$, and we say that the clause $C$ became \textit{unit} and $l$ is the \textit{unit literal}. When a literal $x$ implies a literal $y$ via clause $C$, we say that $x$ \textit{propagates} $y$ because of $C$. When a literal $x$ is \textit{propagated to saturation}, it means that for each clause that contains $x$ and is unit, we derive the unit literal. The \textit{conflict clause} is the clause that becomes completely falsified in the implication graph, and the \textit{learned clause} is the clause learned by the solver following conflict analysis (in these proofs, 1UIP). We use these properties of BCP:

\begin{observation} 
A literal is propagated to saturation before any of the literals that it implies are propagated.
\end{observation}

\begin{observation}
There are no further propagations once a conflict has been identified. 
\end{observation}

We first define some general cases where clauses can become unit. 
\begin{observation}[\bf Transitivity Conditions]
Recall that transitivity clauses are of the form $A(i, j, k) = \neg P_{i, j} \lor \neg P_{j, k} \lor P_{i, k}$. There are 3 cases where the clause could become unit: 
\begin{enumerate}
    \item if you derive $P_{i, j}$ and $P_{j, k}$
    \item if you derive $P_{j, k}$ and $\neg P_{i, k}$
    \item if you derive $P_{i, j}$ and $\neg P_{i, k}$
\end{enumerate}

Additionally, two literals can only (together) cause a single Transitivity clause to become unit (follows simply from the fact that Transitivity clauses only have 3 literals, and no two Transitivity clauses share more than 1 literal). 
\end{observation}

\begin{observation} [\bf Antisymmetry Conditions]
Recall that antisymmetry clauses are of the form $B(i, j) = \neg P_{i, j} \lor \neg P_{j, i}$. The only way this clause becomes unit is by deriving a literal with positive polarity. A literal can only cause a single Antisymmetry clause to become unit (follows simply from the fact that Antisymmetry clauses only have 2 literals, and no two Antisymmetry clauses share more than 1 literal). 
\end{observation}

\begin{observation} [\bf Non-minimality Conditions]
This only becomes unit when there are $n - 2$ literals all sharing the same column and derived with negative polarity.
\end{observation}

\begin{observation} [\bf Head Conditions]
The Head clauses can only become unit when there are $n - 3$ or $n - 4$ literals all with column 1 and derived with negative polarity.
\end{observation}

\begin{observation} [\bf Descending Cascade Conditions]
Descending Cascade clauses can only become unit when literals within the same column (notably not column 1) are decided with negative polarity.
\end{observation}

\begin{observation} [\bf Tail Conditions]
Tail clauses can only become unit when literals with column 1 are decided with negative polarity.
\end{observation}

\begin{observation} [\bf Saturation Condition]
When a literal $l$ is propagated to saturation and the implied literals cause there to be a conflict, none of the literals that $l$ implied are propagated. This is a direct consequence of the BCP observations. 
\end{observation}

\subsection{First Head Clause Completeness}
See Figure \ref{fig:head1} for the implication graph.
\begin{itemize}
    \item Transitivity: Each decision literal causes a Transitivity clause to become unit when $P_{n, 1}$ is propagated to saturation. By the Transitivity Conditions, since two literals together can only make a single Transitivity clause unit, it follows that literals implied by $P_{n, 1}$ must be propagated in order for any other Transitivity clauses to become unit. However, since propagating $P_{n, 1}$ to saturation results in conflict, it follows by the Saturation Condition that no literal implied by $P_{n, 1}$ is propagated. Thus, no other Transitivity clause can become unit. 
    \item Antisymmetry: The only literal with positive polarity is $P_{n, 1}$, and the corresponding antisymmetry clause that becomes unit is shown. By the Antisymmetry Conditions, it follows that no other Antisymmetry clause could become unit. 
    \item Non-minimality: The occurrences of $n - 2$ literals sharing the same column with negative polarity are shown, and by the Non-minimality Conditions it follows that no other Non-minimality clause could become unit. 
\end{itemize}

\subsection{Second Head Clause Completeness}
See Figure \ref{fig:head2} for the implication graph.
\begin{itemize}
    \item Transitivity: By Case 1 of the Transitivity Conditions, $P_{n - 1, 1}$ and $P_{n, n - 1}$ together imply $P_{n, 1}$; this is omitted from the implication graph because $P_{n, 1}$ cannot be propagated by the Saturation Condition and it does not appear in the conflict clause. By the Transitivity Conditions and the Saturation Condition, it follows that no other Transitivity clause could become unit.
    \item Antisymmetry: The only literals with positive polarity have their corresponding antisymmetry clauses that become unit shown. By the Antisymmetry Conditions, it follows that no other Antisymmetry clause could become unit. 
    \item Non-minimality: The occurrences of $n - 2$ literals sharing the same column with negative polarity are shown, and by the Non-minimality Conditions it follows that no other Non-minimality clause could become unit. 
    \item First Head Clause: This does become unit and is shown.
\end{itemize}

\subsection{First Descending Cascade Clause Completeness}
See Figure \ref{fig:dcas1} for the implication graph.
\begin{itemize}
    \item By Case 1 of the Transitivity Conditions, $P_{n - 2, 1}$ and $P_{n, n - 2}$ together imply $P_{n, 1}$; this is omitted from the implication graph because $P_{n, 1}$ cannot be propagated by the Saturation Condition and it does not appear in the conflict clause. By the Transitivity Conditions and the Saturation Condition, it follows that no other Transitivity clause could become unit. 
    \item Antisymmetry: The only literals with positive polarity have their corresponding antisymmetry clauses that become unit shown. By the Antisymmetry Conditions, it follows that no other Antisymmetry clause could become unit. 
    \item Non-minimality: The occurrences of $n - 2$ literals sharing the same column with negative polarity are shown, and by the Non-minimality Conditions it follows that no other Non-minimality clause could become unit. 
    \item Head clauses: The second Head clause becomes unit and is shown. The first Head clause cannot become unit if the second Head clause becomes unit. 
\end{itemize}

\subsection{Cascade Lemma Completeness}
See Figures \ref{fig:dcas2}, \ref{fig:dcas3}, \ref{fig:dcas4}, \ref{fig:dcas5} for the implication graph.
\begin{itemize}
    \item Transitivity: 
    \begin{itemize}
        \item Case 1 ($k \neq l, k = n - 2$): Shared row/col with positive polarity (Case 1 of the Transitivity Conditions) has one possible occurrence which would propagate the literal $P_{n, l}$ as a result of $A(k + 2, k + 1, l)$; this is omitted from the implication graph due to the Saturation Condition. Case 2 of the Transitivity Conditions has all of its occurrences shown. Case 3 of the Transitivity Conditions does not occur.
        \item Case 2 ($k = l, k \neq n - 2)$: It can easily be seen that there is no other possible transitivity clause that could become unit.
        \item Case 3 ($l < k < n - 2$): It is clear that this is analogous to Case 2.
        \item Case 4 ($k = l = n - 2$: It is clear that this is analogous to Case 1.
    \end{itemize}
    \item Antisymmetry: The only literals with positive polarity have their corresponding antisymmetry clauses that become unit shown. By the Antisymmetry Conditions, it follows that no other Antisymmetry clause could become unit. 
    \item Non-minimality: In all cases this cannot occur by the Non-minimality Conditions. 
    \item Head clauses: We never branch or derive any literals with column $1$ so this cannot occur by the Head Conditions.
    \item Descending Cascade Clauses: 
    \begin{itemize}
        \item Case 1: We derive literals with negative polarity in column $l$, $k + 1$, and $k + 2$. All of the columns have their corresponding Descending Cascade clauses shown to be unit. By the Descending Cascade Conditions, no other Descending Cascade clauses can become unit. 
        \item Case 2: We derive literals with negative polarity in columns $l$ and $k + 1$, and the implication graph shows both of the corresponding Descending Cascade clauses become unit. By the Descending Cascade Conditions, no other Descending Cascade clauses can become unit. 
        \item Case 3: Analogous to Case 2. 
        \item Case 4: Analogous to Case 1.
    \end{itemize}
\end{itemize}

\subsection{Descending Lemma Completeness}
See Figure \ref{fig:dcas6} for the implication graph.

Note that $P_{n, l}$ is not propagated due to the Saturation Condition. 
\begin{itemize}
    \item Transitivity: Shared row/col with positive polarity (Case 1 of the Transitivity Conditions) has one occurrence and is shown. Same column with opposite polarity (Case 2 of the Transitivity Conditions) also has all of its occurrences shown. Same row with opposite polarity (Case 3 of the Transitivity Conditions) does not occur. 
    \item Antisymmetry: The only literals with positive polarity already have their corresponding antisymmetry clauses that become unit shown. By the Antisymmetry Conditions, it follows that no other Antisymmetry clause could become unit. 
    \item Non-minimality: The only occurrences of this are shown, where $D(l - 1)$ becomes unit and $D(n)$ is the conflict clause. By the Non-minimality Conditions, no other Non-minimality clause could become unit. 
    \item Head clauses: We never branch or derive any literals with column $1$ so this cannot occur by the Head Conditions. 
    \item Descending Cascade Clauses: The columns with literals branched on or propagated with negative polarity are $l$, $l - 1$, and $n$. The Non-minimality clauses $D(l - 1)$ and $D(n)$ become unit, so Descending Cascade clauses with literals in those columns cannot become unit. The Descending Cascade clause $C(l, l - 1)$ is shown to become unit. By the Descending Cascade Conditions, no other Descending Cascade clause can become unit.
\end{itemize}

\subsection{Pivot Lemma Completeness}
See Figure \ref{fig:pivot} for the implication graph.

\begin{itemize}
    \item Transitivity: Case 1 of the Transitivity Conditions occurs with $P_{1, 2}$ and $P_{n - 2, 1}$ to imply $P_{n - 2, 2}$, but is omitted from the implication graph since $P_{n - 2, 2}$ will not be propagated due to the Saturation Condition and does not appear in the conflict clause. By the Transitivity Conditions and the Saturation Condition, no other Transitivity clause can become unit. 
    \item Antisymmetry: The only literals with positive polarity already have their corresponding antisymmetry clauses that become unit shown. By the Antisymmetry Conditions, it follows that no other Antisymmetry clause could become unit. 
    \item Non-minimality: By the Non-minimality Conditions, no Non-minimality clauses can become unit. 
    \item Head clauses: The second Head clause becomes unit and is shown. The first Head clause cannot become unit. 
    \item Descending Cascade Clauses: The columns with literals branched on or propagated with negative polarity are $2$, $1$, and $n - 2$. We don't consider $1$ for the Descending Cascade; for column 2, since we already learned the clause $P_{1, 2}$, all of the Descending Cascade Clauses with column 2 are already satisfied and cannot become unit. We falsify a Descending Cascade clause associated with column $n - 2$ which is shown. By the Descending Cascade Conditions, no other Descending Cascade clauses can become unit. 
\end{itemize}

\subsection{Tail Lemma Completeness}
See Figure \ref{fig:tail} for the implication graph.

\begin{itemize}
    \item Transitivity:  Case 1 of the Transitivity Conditions is only possible when we see literals such that one's column is the same as the other's row and they are derived with the same polarity, which happens with $P_{1, 2}$ and $P_{k + 1, 1}$, from which we get $P_{k + 1, 2}$ as shown. Case 2 of the Transitivity Conditions cannot occur apart from what is shown in the implication graph, because we never derive any other literals with the same column and opposite polarity. Finally, Case 3 of the Transitivity Conditions is only possible when we see literals such that they have the same row and they are derived with opposite polarity, which happens with $P_{1, 2}$ and $\neg P_{1, k + 1}$; however, $\neg P_{1, k + 1}$ cannot be propagated due to the Saturation Condition.
    \item Antisymmetry: The implication graph already shows all occurrences of this; $P_{k + 1, 2}$ is not propagated by the Saturation Condition. By the Antisymmetry Conditions, no other Antisymmetry clauses can become unit. 
    \item Non-minimality: No Non-minimality clauses can become unit by the Non-minimality Conditions. 
    \item Head: No Head clauses can become unit by the Head Conditions. 
    \item Descending Cascade: The only candidate is column $k + 1$, which is the conflict clause in the implication graph. No other Descending Cascade clauses can become unit by the Descending Cascade Conditions. 
    \item Tail: The only tail clause which becomes unit is shown in the implication graph; no other tail clauses can become unit by the Tail Conditions. 
\end{itemize}

\end{document}